\documentclass[12pt,a4paper,reqno]{article}
\usepackage{amsmath,amsthm,amsfonts,amssymb,bbm}
\usepackage{graphicx,psfrag,subfigure}
\usepackage{cite}
\usepackage{hyperref}

\numberwithin{equation}{section}
\numberwithin{figure}{section}

\newcommand{\Id}{\mathbbm{1}}
\newcommand{\Or}{\mathcal{O}}
\newcommand{\Pb}{\mathbbm{P}}

\newcommand{\e}{\varepsilon}
\newcommand{\I}{{\rm i}}

\newcommand{\R}{\mathbb{R}}
\newcommand{\N}{\mathbb{N}}
\newcommand{\Z}{\mathbb{Z}}

\renewcommand{\Re}{\mathrm{Re}}
\renewcommand{\Im}{\mathrm{Im}}

\DeclareMathOperator*{\Ai}{Ai}

\newtheorem{prop}{Proposition}[section]
\newtheorem{thm}[prop]{Theorem}
\newtheorem{lem}[prop]{Lemma}
\newtheorem{defin}[prop]{Definition}

\newtheorem{cla}[prop]{Claim}

\newtheorem{rem}[prop]{Remark}
\newenvironment{remark}{\begin{rem}\normalfont}{\end{rem}}

\title{Limit processes for TASEP\\ with shocks and rarefaction fans}

\author{Ivan Corwin\thanks{Courant Institute of Mathematical Sciences,
  New York University, \newline 251 Mercer Street, New York, NY 10012, USA; E-mail:~\texttt{corwin@cims.nyu.edu}},
Patrik L. Ferrari\thanks{Institute for Applied Mathematics, University of Bonn, Endenicher Allee 60,\newline 53115 Bonn, Germany; E-mail:~\texttt{ferrari@uni-bonn.de}},
Sandrine P\'ech\'e\thanks{Institut Fourier, 100 Rue des maths, 38402 Saint Martin d'Heres, France;\newline E-mail:~\texttt{Sandrine.Peche@ujf-grenoble.fr}}}

\date{11. May 2010}

\begin{document}
\sloppy
\maketitle

\begin{abstract}
We consider the totally asymmetric simple exclusion process (TASEP) with two-sided Bernoulli initial condition, i.e., with left density $\rho_-$ and right density $\rho_+$. We study the associated height function, whose discrete gradient is given by the particle occurrences. Macroscopically one has a deterministic limit shape with a shock or a rarefaction fan depending on the values of $\rho_\pm$. We characterize the large time scaling limit of the multipoint fluctuations as a function of the densities $\rho_\pm$ and of the different macroscopic regions. Moreover, using a slow decorrelation phenomena, the results are extended from fixed time to the whole space-time, except along the some directions (the characteristic solutions of the related Burgers equation) where the problem is still open.

On the way to proving the results for TASEP, we obtain the limit processes for the fluctuations in a class of corner growth processes with external sources, of equivalently for the last passage time in a directed percolation model with two-sided boundary conditions. Additionally, we provide analogous results for eigenvalues of perturbed complex Wishart (sample covariance) matrices.
\end{abstract}

\section{Introduction}
We consider the totally asymmetric simple exclusion process (TASEP) on $\Z$. This is one of the basic one-dimensional interacting stochastic particle
systems that, despite its simplicity, exhibits a number of interesting features. TASEP is a Markov process $\eta_t$ with state space $\{0,1\}^{\Z}$. For a given time $t\in \R_+$ and position $x\in \Z$, we say that site $x$ is occupied at time $t$ if  $\eta_t(x)=1$ and it is empty if $\eta_t(x)=0$ (we can have at most one particle at each site: exclusion principle).
The dynamics is defined as follows. Particles jump to the neighboring right site with rate $1$ provided that the site is empty. Jumps are independent of each other and take place after an exponential waiting time with mean $1$, which is counted from the time instant when the right neighbor site is empty (for a rigorous construction, see~\cite{Li99,Li85b}).

One may study a variety of slightly different observables of TASEP such as the total current, the location of a tagged particles or the TASEP height function. Here we focus on the height function, $h_t$, defined from a TASEP configuration $\eta_{t}$ as
\begin{equation}\label{height_function}
 h_t(j) = \begin{cases}
       2N_t + \sum_{i=1}^{j} (1-2\eta_t(i)) & \textrm{for }j\geq 1,\\
       2N_t & \textrm{for } j=0,\\
       2N_t - \sum_{i=j+1}^{0} (1-2\eta_t(i)) & \textrm{for }j\leq -1,
       \end{cases}
\end{equation}
where $N_t$ is the total number of particles which jumped from site $0$ to site $1$ during the time interval $[0,t]$.

In this paper we consider the simplest family of (random) initial condition, in which shockwaves or rarefaction fans occur. More precisely, our initial condition is Bernoulli product measure with density $\rho_-$ on $\Z_-=\{\ldots,-2,-1\}$ and $\rho_+$ on $\Z_+^*=\{0,1,\ldots\}$. We refer to this as \emph{two-sided Bernoulli initial condition}. Particular cases which have already been studied are:\\[0.5em]
$\bullet$ the step-initial condition ($\rho_-=1$ and $\rho_+=0$), where $\Z_-$ is completely filled. In this case, there is a rarefaction fan, the fluctuations of $h_t$ scale as $t^{1/3}$, the correlation length as $t^{2/3}$, and the limit process is the Airy$_2$ process (see the case $b\equiv 0$ in~\cite{BF07}). The Airy$_2$ process occurred first in closely related growth models~\cite{PS02,Jo03b}.\\[0.5em]
$\bullet$ Stationary initial condition ($\rho\equiv\rho_-=\rho_+\in (0,1)$). The only stationary and translation invariant measures are Bernoulli product measures with constant density $\rho\in [0,1]$ ($\rho=0$ and $\rho=1$ are however trivial)~\cite{Lig76}. The scaling limit for the multi-point distribution of stationary TASEP has been recently unraveled in~\cite{BFP09}.\vspace{0.5em}

Therefore we have only to focus on $\rho_+\neq\rho_-$ (the results below are the content of Theorem~\ref{ThmTASEP} and are illustrated in Figure~\ref{FigMacroTASEP}). There are two cases:\\[0.5em]
(a) $\rho_- > \rho_+$. For large time $t$ the asymptotic density decreases linearly from $1-\rho_-$ to $1-\rho_+$ over the region from $(1-2\rho_-)t$ to $(1-2\rho_+)t$ called a rarefaction fan (see Figure~\ref{FigMacroTASEP}~(a)). In this region the height fluctuations live on a $t^{1/3}$ and are governed by the Airy$_2$ process like for step-initial condition (with correlation length scaling as $t^{2/3}$). Around positions $(1-2\rho_\pm)t$ the randomness of the initial conditions start being relevant and there is a transition process from Airy$_2$ to Brownian Motion. When the fluctuations coming from the initial condition are on the $t^{1/2}$ scale, they dominate the fluctuations created by the dynamics ($t^{1/3}$ scale) and are governed by Brownian Motion. This is the case on the left and on the right of the rarefaction fan.\\[0.5em]
(b) $\rho_- < \rho_+$. For large time $t$ there is a macroscopic shock with density jump from $\rho_-$ to $\rho_+$ around the position $(1-\rho_- - \rho_+)t$ (see Figure~\ref{FigMacroTASEP}~(b)). For large time $t$, the fluctuations on the left and on the right of the shock are independent. Of particular interest, is then the joint-distribution of the height function around the shockwave $(1-\rho_- - \rho_+)t$ at different times.
The initial conditions considered here are random and therefore looking far enough away the initial randomness becomes more important than the fluctuations created by the dynamics, which live on the $t^{1/3}$ scale only. For non-random initial conditions this does not happen and further limit processes arise, see~\cite{BFPS06,BFS07} and, for one-sided random initial condition, see~\cite{BFS09}.\\[.5em]
The results for the one-point distributions in Theorem~\ref{ThmTASEP} were conjectured in~\cite{PS01} and recently proven in~\cite{BC09}. The conjecture was based on universality, since analogue results were available for a stochastic growth model (the polynuclear growth (PNG) model)~\cite{BR00}, which is in the same universality class, named for Kardar-Parisi-Zhang (KPZ)~\cite{KPZ86}. The extension to multi-point distributions at fixed time in the PNG model was carried out in~\cite{SI04} (except for the case corresponding to stationary TASEP).

Extensions away from fixed-time have been previously obtained in TASEP (with different type of initial conditions) in~\cite{SI07,BF07}. However, the extension was technically restricted to \emph{space-like paths}, for which one could still get explicit expressions for the correlation functions. Our main result (Theorem~\ref{ThmTASEP}) is much more general, since it covers almost all space-time. In particular we can analyze situations where the correlation functions are not explicitly known!

The only directions where the question of the limit process remains open are the characteristic solutions of the Burgers equation associated to TASEP, also called \emph{characteristic lines}~\cite{Var04,Ev98}. Along these space-time lines the appropriate scaling limit is different, because the decorrelation occurs on a much longer time scale compared with the usual decorrelation length. This because second-class particles follow (on a macroscopic scale) exactly these trajectories. This phenomenon was first proven in a PNG model~\cite{Fer08} and it is called \emph{slow decorrelation phenomenon} (recently proven in greater generality in~\cite{CFP10b}, see Proposition~\ref{ThmSlowDec} below for TASEP).

\bigskip
\noindent \textbf{Methods in the proof of Theorem~\ref{ThmTASEP}}\\
The proof of our main result, Theorem~\ref{ThmTASEP}, employs a combination of many of the state-of-the-art methods in the study of TASEP fluctuations (exact determinantal correlation formulas, the connection to last passage percolation, coupling methods, and slow decorrelation). In outlining our proof we also provide a brief review of the literature on these different techniques.

The first step in the proof is to establish a multipoint fluctuation result along a fixed space-time cut for a simpler initial condition corresponding to fixing $\rho_+=0$. For certain space-time cuts there exist exact determinantal expressions for the correlation functions. Specifically, if one considers the fixed-time cut for TASEP (i.e., the joint-distribution of the height at a fixed time) then there is a way to extend~\cite{BFPS06,BFS09} to get the necessary formulas for the correlation functions~\cite{Bor08_privatecomm}. We consider a different cut which corresponds to a directed last passage percolation model with one-sided boundary condition (see Section~\ref{subsectOneSided}). In that case the Schur process gives the multipoint correlation kernel, Equation~(\ref{cor_kernel}) (for details on the Schur process and applications see~\cite{Ok01,OR01,Jo05,BP07}). In the TASEP setting, the Schur process is the process of a given (tagged) particle observed at different times. Using techniques of asymptotic analysis for Fredholm determinants we can extract our desired limit theorems from these formulas. This is done in Section~\ref{subsectProofOneSided} and recorded as Proposition \ref{ThmOneSidedLPPbasic}. As opposed to the related work of~\cite{SI04}, this is the only case for which we must appeal to the exact correlation formulas and take asymptotics.

From this point on our proof relies entirely on probabilistic methods. The only case which is not covered by these other techniques is $\rho_-=\rho_+$ at the characteristic speed, but this was analyzed independently in~\cite{BFP09}. These methods allow us to avoid the more involved shift and analytic continuation arguments of~\cite{BFP09} in addition to the Schur process (see Remark~9 of~\cite{BFP09}).

The first probabilistic method we use is slow decorrelation~\cite{CFP10b} (given as Proposition~\ref{ThmSlowDec}). This implies that the fluctuation limit process for one-sided initial conditions extends away from cut on which it was proved (Proposition~\ref{ThmOneSidedLPP}). In order to bootstrap the one-sided process result to the full two-sided case we appeal to a coupling method introduced in~\cite{BC09}. In fact, versions of this method can be found in the literature in~\cite{AD95,Sep98,Rez02} under the names ``microscopic Lax-Oleinik formula'' or ``strong monotonicity''. The new component offered by the coupling method of~\cite{BC09} is that due to a few (fairly simple) lemmas (see Section~\ref{coupling_lemmas}) one can now prove $t^{1/3}$ fluctuation results. The essential idea behind these methods is that statistics associated to complicated particle system or growth process initial conditions can often be written in terms of statistics associated to simpler particle systems which have been coupled to the original system. Once this connection is in place (for us Section~\ref{proof_two_sided_section}) it is generally possible to translate asymptotic results about the simplier systems into results about the original (more complicated) system. This method is applied in Section~\ref{proof_two_sided_section} to prove Theorem~\ref{ThmTwoSidedLPP} which is the last passage percolation equivalent of our main TASEP result. Finally, in Section~\ref{subsectProofTASEP} we show how to translate this back into a proof of Theorem~\ref{ThmTASEP}.

One-side last passage percolation is closely related to the largest eigenvalue of some complex Wishart (sample covariance) matrices~\cite{BBP06}. Using the connection established in~\cite{BP07,DW08}, we restate our one-sided last passage percolation process result in terms of a random matrix eigenvalue process(see Theorem~\ref{ThmRM}).

\medskip
There are a variety of conjectured results which go under the title of universality. The results of this paper deal with universality of the PNG and continuous time TASEP. However, TASEP is also the extreme case of the partially asymmetric version (PASEP), where particles can jump both left and right with different jump rates. For the one-point distribution function progress in this direction was made in~\cite{FKS91, Fer90,  FF94, FF94b, FK95} in the early 1990s. Very recently, due to the efforts of Tracy and Widom~\cite{TW08c, TW08, TW08b, TW09, TW09b}, Derrida and Gerschenfeld~\cite{DG09}, Bal\'{a}zs and Sepp\"{a}l\"{a}inen~\cite{BS06, BS09}, Quastel and Valk\'{o}~\cite{QV07}, Mountford and Guiol~\cite{MG05} significant progress has been made in answering this question in the general PASEP. Of particular note is the recent result of Tracy and Widom~\cite{TW09b} which shows that the results of~\cite{BC09} for TASEP with two-sided Bernoulli initial conditions extend to the PASEP setting for $\rho_+=0$ and general values of $\rho_-$. It seems hopeful that the integrable systems methods which proved useful in that paper will, eventually be able to deal with general two-sided Bernoulli initial conditions as well as multi-point distribution functions. With that eventuality in mind, this paper should serve as a guide in that pursuit.

In addition to extending TASEP results to the context of the PASEP, Tracy and Widom's formula has played a prominent role in the long sought after calculation~\cite{ACQ10,SS10} of the one-point function for the KPZ stochastic PDE with narrow wedge initial condition (for an alternative approach using the replica trick see also~\cite{CDR10,Dot10}).

\subsection*{Acknowledgments}
We wish to thank Jinho Baik for helpful discussions related to this material. Some of these discussions occurred during the MSRI Random Matrix Theory workshop during the summer of 2009. I.~Corwin would like to thanks G\'erard Ben~Arous for introducing him to the study of TASEP fluctuations. His work is partially funded by the NSF Graduate Research Fellowship and also has received travel funding from the PIRE grant OISE-07-30136. S.~P\'ech\'e would like to thank Herv\'e Guiol for useful discussion on TASEP and her work is partially supported by the Agence Nationale de la Recherche grant ANR-08-BLAN-0311-01.

\section{Results}\label{sectResults}
Here we present the limit results, first for TASEP, then for last passage percolation and we end with random matrices. The limit processes in the following statements are defined in Section~\ref{subsectDefProc}.

\subsection{Continuous time TASEP}\label{subsectTASEP}
We want to analyze the fluctuations of the height function (\ref{height_function}) with respect to the macroscopic behavior. Thus, the first quantity we need to determine is the limit shape
\begin{equation}
h_{\rm ma}(\xi):=\lim_{t\to\infty}\frac{1}{t}h_t(\lfloor \xi t\rfloor)
\end{equation}
which can be obtained by integrating the asymptotic macroscopic density of particles, $\varrho(\xi,\tau)$, given heuristically by
\begin{equation}
\varrho(\xi,\tau):=\lim_{T\to\infty}\Pb(\textrm{there is a particle at }[\xi T]\textrm{ at time }\tau T).
\end{equation}
The average current of particles for a density $\varrho$ is $\varrho(1-\varrho)$, thus $\varrho$ satisfies Burgers equation~\cite{Rez91}
\begin{equation}
\partial_\tau\varrho+\partial_\xi(\varrho(1-\varrho))=0.
\end{equation}
The initial condition $\varrho(\xi,0)=\rho_-$ for $\xi<0$ and $\varrho(\xi,0)=\rho_+$ for $\xi>0$ gives:\\
(a) for $\rho_- \geq \rho_+$,
\begin{equation}
\varrho(\xi,1)=\begin{cases}
\rho_-    &\textrm{for }\xi\leq 1-2\rho_-,\\
(1-\xi)/2 &\textrm{for }\xi\in [1-2\rho_-,1-2\rho_+],\\
\rho_+    &\textrm{for }\xi\geq 1-2\rho_+,
\end{cases}
\end{equation}
(b) while for $\rho_- < \rho_+$,
\begin{equation}\label{eqn2.5}
\varrho(\xi,1)=\begin{cases}
\rho_-    &\textrm{for }\xi< 1-(\rho_- + \rho_+),\\
\rho_+    &\textrm{for }\xi> 1-(\rho_- + \rho_+).
\end{cases}
\end{equation}
The characteristic lines\footnote{These characteristics are the ones coming from the entropy condition~\cite{Var04,Ev98}.}, $(t,x(t))_{t\geq 0}$, of the Burgers equation with constant density $\rho$ are straight lines with speed $1-2\rho$: $\{x(t)-x(0)=(1-2\rho)t, t\geq 0\}$. In the case of non-constant density (see case~(a)), then all the rays leaving from the origin with speed $\xi\in [1-2\rho_-,1-2\rho_+]$ are also characteristic lines (see Figure~\ref{FigTwoSided}).

Translated into the limit shape using (\ref{height_function}), one obtains:\\
(a) for $\rho_- \geq \rho_+$,
\begin{equation}\label{eq2.6}
h_{\rm ma}(\xi)=\begin{cases}
2\rho_-(1-\rho_-)+(1-2\rho_-)\xi &\textrm{for }\xi\leq 1-2\rho_-,\\
(1+\xi^2)/2 &\textrm{for }\xi\in [1-2\rho_-,1-2\rho_+],\\
2\rho_+(1-\rho_+)+(1-2\rho_+)\xi    &\textrm{for }\xi\geq 1-2\rho_+,
\end{cases}
\end{equation}
(b) while for $\rho_- < \rho_+$,
\begin{equation}
h_{\rm ma}(\xi)=\begin{cases}
2\rho_-(1-\rho_-)+(1-2\rho_-)\xi    &\textrm{for }\xi< 1-(\rho_- + \rho_+),\\
2\rho_+(1-\rho_+)+(1-2\rho_+)\xi    &\textrm{for }\xi> 1-(\rho_- + \rho_+).
\end{cases}
\end{equation}
This is illustrated in Figure~\ref{FigMacroTASEP}.
\begin{figure}[t!]
\begin{center}
\psfrag{(a)}[c]{(a)}
\psfrag{(b)}[c]{(b)}
\psfrag{x}[c]{$\xi$}
\psfrag{rho}[l]{$\varrho(\xi,1)$}
\psfrag{rho-}[l]{$\rho_-$}
\psfrag{rho+}[l]{$\rho_+$}
\psfrag{hma}[l]{$h_{\rm ma}(\xi)$}
\psfrag{xim}[c]{$\xi_-$}
\psfrag{xip}[c]{$\xi_+$}
\psfrag{xis}[c]{$\xi_s$}
\psfrag{B}[c]{$\mathcal{B}$}
\psfrag{B'}[c]{$\mathcal{B}'$}
\psfrag{A2}[l]{$\mathcal{A}_2$}
\psfrag{BM2}[l]{$\mathcal{A}_{{\rm BM}\to 2}$}
\psfrag{BM2'}[l]{$\mathcal{A}_{2\to{\rm BM}}$}
\includegraphics[width=12cm]{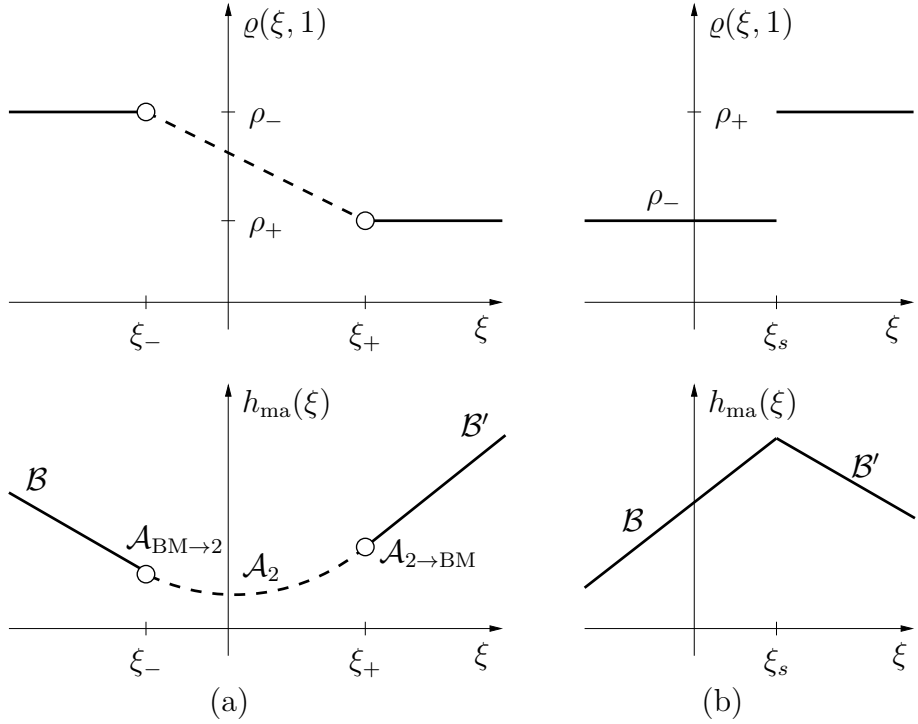}
\caption{The asymptotic density $\varrho$ and the limit shape in the cases (a) $\rho_->\rho_+$ and (b) $\rho_-<\rho_+$. Transitions happen at $\xi_\pm=1-2\rho_\pm$ and shockwave at $\xi_s=1-(\rho_-+\rho_+)$. The fluctuations processes are also indicated: $\mathcal{B}$ and $\mathcal{B}'$ two independent standard Brownian motions, $\mathcal{A}_2$ the Airy$_2$ process, $\mathcal{A}_{{\rm BM}\to 2}$ is the transition process from Brownian behavior to Airy$_2$ process, and $\mathcal{A}_{2\to{\rm BM}}$ is its time-reversed version. See Section~\ref{subsectDefProc} for definitions.}
\label{FigMacroTASEP}
\end{center}
\end{figure}

For simplicity, we discuss the fixed time fluctuation results of our main Theorem~\ref{ThmTASEP}, which is however more general and holds for unequal times too. Concerning the fluctuations for $\rho_-\neq \rho_+$, if we focus around a macroscopic position $\xi t$ we have:\\[0.5em]
Case (a) and $1-2\rho_-<\xi<1-2\rho_+$: the limit shape is curved and the behavior is like the one of step-initial condition, namely, for large time $t$ the fluctuations scale as $t^{1/3}$, correlations as $t^{2/3}$, and the multi-point statistics are governed by the Airy$_2$ process, $\mathcal{A}_2$. More precisely, there are two coefficients $\kappa_h=(2(1-\xi)^2)^{1/3}$ and $\kappa_v=-(1-\xi^2)^{2/3}/2^{1/3}$ depending only on $\xi$ (compare with (\ref{eq2.8}) below) such that for large time $t$,
\begin{equation}\label{eq2.8b}
h(\xi t+\tau \kappa_h t^{2/3})\simeq t h_{\rm ma}(\xi+\tau \kappa_h t^{-1/3})+\kappa_v {\cal A}_2(\tau) t^{1/3}.
\end{equation}
Case (a) and $\xi=1-2\rho_-$ (or $\xi=1-2\rho_+$): the influence of the randomness in the initial condition and the randomness built up by the dynamics are of the same order. The fluctuations of $h_t$ are on the $t^{1/3}$ scale with correlation scale $t^{2/3}$ and are governed by a transition process $\mathcal{A}_{{\rm BM}\to 2}$ between Brownian Motion behavior and the Airy$_2$ process.\\[0.5em]
Case (a) and $\xi<1-2\rho_-$ (or $\xi>1-2\rho_+$) or Case (b) away from the shock position: the influence of the initial randomness dominates and one has simply Brownian Motion (with fluctuations scale $t^{1/2}$ and correlation scale $t$). The two sides are asymptotically independent. \\[0.5em]
Case (b) at the shock position: the statistics of the height function is influenced by both the right and left particle densities.\\[0.5em]
The stationary case, $\rho_-=\rho_+=\rho$ was analyzed already in~\cite{BFP09} (see Theorem~1.7 therein). At $\xi=1-2\rho$ the fluctuations are of order $t^{1/3}$ and there is a transition process ${\cal A}_{\rm stat}$ over a distance of order $t^{2/3}$ to the Gaussian behavior.

In a related model~\cite{Fer08} the following slow decorrelation phenomenon was noticed: along the characteristic lines the height-height correlations live on a longer space-time scale than the fixed-time correlation scale. For instance, in the rarefaction fan, the height function at two space-time points on the same characteristic line, the first at time $T$ and the second at time $T+T^\nu$, $\nu<1$, will differ by a deterministic factor (speed of growth $\times$ $T^\nu$) plus $o(T^{1/3})$. This means that the two height functions (centered and rescaled by $T^{-1/3}$) are asymptotically the same random variable (i.e., they are perfectly correlated on the $T^{1/3}$ scale). The proof in~\cite{Fer08} uses several results of other papers and a considerable amount of work is needed to reproduce them for other models. While looking for a proof for the TASEP, we discovered a much simpler proof, which applies not only to TASEP but to a large number of models in the KPZ class, see~\cite{CFP10b}. The statement for TASEP is reported in Proposition~\ref{ThmSlowDec}.

This allows us to extend the fixed-time statement to space-time. This is the reason for the following limit theorem: for $\xi\in [1-2\rho_-,1-2\rho_+]$, let us set
\begin{equation}\label{eq2.8}
\begin{aligned}
X(\tau,\theta)&=\lfloor\xi(T+\theta T^\nu)+\tau (2(1-\xi^2))^{1/3} T^{2/3}\rfloor,\\
H(\tau,\theta,s)&=\frac{1+\xi^2}{2}(T+\theta T^\nu)+\xi\tau (2(1-\xi^2))^{1/3} T^{2/3}+(\tau^2-s) \frac{(1-\xi^2)^{2/3}}{2^{1/3}} T^{1/3}.
\end{aligned}
\end{equation}
The value of $H(\tau,\theta,0)$ is a generalization of the term $t h_{\rm ma}(\xi+\tau \kappa_h t^{-1/3})$ in (\ref{eq2.8b}), namely the macroscopic approximation. Indeed,
\begin{equation}
H(\tau,\theta,0)=(T+\theta T^\nu) \frac12 \bigg(1+\left(\frac{X(\tau,\theta)}{T+\theta T^\nu}\right)^2\bigg)+o(T^{1/3}),
\end{equation}
compare with (\ref{eq2.6}), while $H(\tau,\theta,s)-H(\tau,\theta,0)$ measures the fluctuations. The definitions of the limit processes occurring in the following theorem are collected in Section~\ref{subsectDefProc}.

\begin{thm}\label{ThmTASEP}
(a) Fix $m\in \N$, $\nu\in [0,1)$, $\xi\in\R$, and $\rho_+\in (0,1]$, $\rho_-\in[0,1)$.
Then, for any choice of real numbers $\tau_1<\tau_2<\ldots,\tau_m$, $\theta_1,\ldots,\theta_m$, and $s_1,\ldots,s_m$, we have:\\
(a1) If $\rho_+<\rho_-$ and $\xi\in (1-2\rho_-,1-2\rho_+)$, then
\begin{equation}
\lim_{T\to\infty} \Pb\left(\bigcap_{k=1}^{m} \{h_{T+\theta_k T^\nu}(X(\tau_k,\theta_k))\geq H(\tau_k,\theta_k,s_k)\}\right) = \Pb\left(\bigcap_{k=1}^{m}\{\mathcal{A}_2(\tau_k)\leq s_k\}\right).
\end{equation}
(a2) If $\rho_+<\rho_-$ and $\xi=1-2\rho_-$, then
\begin{equation}
\lim_{T\to\infty} \Pb\left(\bigcap_{k=1}^{m} \{h_{T+\theta_k T^\nu}(X(\tau_k,\theta_k))\geq H(\tau_k,\theta_k,s_k)\}\right) = \Pb\left(\bigcap_{k=1}^{m}\{\mathcal{A}_{{\rm BM}\to 2}(\tau_k)\leq s_k\}\right).
\end{equation}
(a3) If $\rho_+=\rho_-\equiv \rho$ and $\xi=1-2\rho$, then
\begin{equation}\label{eq2.11}
\lim_{T\to\infty} \Pb\left(\bigcap_{k=1}^{m} \{h_{T+\theta_k T^\nu}(X(\tau_k,\theta_k))\geq H(\tau_k,\theta_k,s_k)\}\right) = \Pb\left(\bigcap_{k=1}^{m}\{\mathcal{A}_{\rm stat}(\tau_k)\leq s_k+\tau_k^2\}\right).
\end{equation}
(b) Fix $m\in \N$, and $\rho_+\in (0,1]$, $\rho_-\in[0,1)$. Fix $m_l,m_s,m_r\in \Z_+^*$ and set $m=m_l+m_s+m_r$ and real numbers $\theta_1,\ldots, \theta_m$. Consider a set of $m$ space-time points with macroscopic coordinates $(\xi_i \theta_i T,\theta_i T)$. Let
$m_l$ of the points be such that $\xi_i<1-\rho_- - \max\{\rho_+,\rho_-\}$, $m_b$ such that $\xi_i>1-\rho_+ - \min\{\rho_+,\rho_-\}$, and $m_s$ on the shockwave ($\xi_i=1-\rho_+-\rho_-$ if $\rho_+>\rho_-$). Then
\begin{equation}\label{eq2.12}
\begin{aligned}
\lim_{T\to \infty} &\Pb\left(\bigcap_{k=1}^{m} \{h_{\theta_k T}(\xi_k \theta_k T)\geq h_{\rm ma}(\xi_k)\theta_k T-2 s_k T^{1/2}\}\right)\\
=& \Pb\left(\bigcap_{k=1}^{m_l+m_s}\left\{\mathcal{B}\left(\theta_k(1-2\rho_--\xi_k)(\rho_-(1-\rho_-))\right)\leq s_k\right\}\right) \\
\times &\Pb\left(\bigcap_{k=m_l+1}^m\left\{\mathcal{B}'\left(\theta_k(\xi_k+2\rho_+-1)(\rho_+(1-\rho_+)\right)\leq s_k\right\}\right).
\end{aligned}
\end{equation}
where $\mathcal{B}$ and $\mathcal{B}'$ are two independent copies of Brownian Motion.
\end{thm}

\begin{remark}
Although in the statement we fix $\theta_1,\ldots,\theta_m$, the same holds true if they depend on $T$ provided that they are uniformly bounded in $T$. What we need is that there exists a $\nu<1$ such that $\lim_{T\to\infty}\ln(|\theta_k T^\nu|)/\ln(T)<1$. For instance, we can take $\theta_k T^\nu=\tilde \theta_k T^{2/3}$ with $\tilde\theta_k$ fixed real numbers.
\end{remark}

\begin{remark}
The case $\rho_+<\rho_-$ and $\xi=1-2\rho_+$ can be recovered from (a2) by particle-hole symmetry. The entries in the Brownian motions in (\ref{eq2.12}) are (proportional to) the projections of the space-time points to time $t=0$ along the characteristics to the initial conditions; the proportionality takes just into account the variance of the random walk of the initial condition. This is illustrated in Figure~\ref{FigTwoSided}.
\end{remark}

\begin{figure}[t!]
\begin{center}
\psfrag{(a)}[c]{(a)}
\psfrag{(b)}[c]{(b)}
\psfrag{x}[c]{$x$}
\psfrag{t}[c]{$t$}
\includegraphics[height=5cm]{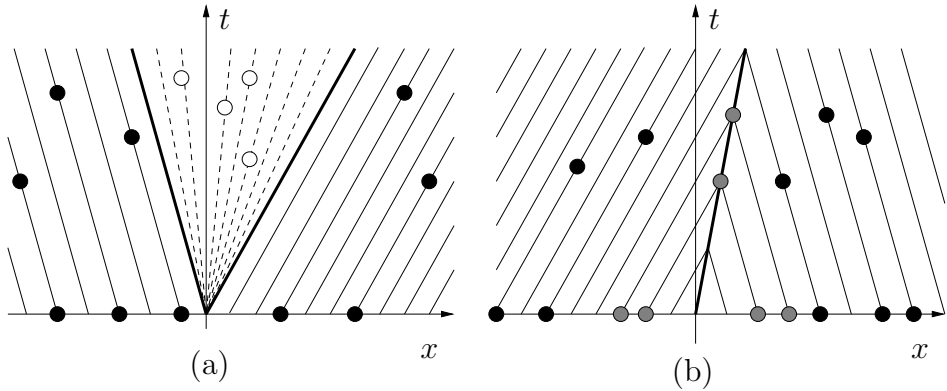}
\caption{Illustration of the characteristic lines for (a) $\rho_->\rho_+$ and (b) $\rho_-<\rho_+$. The fluctuations of the black points depend on fluctuation of their projections to the $t=0$ line. Points on the shockwave, the gray points, depend on the projections on the two directions. Finally, the fluctuations in the rarefaction fan (white points) do not depend on the initial randomness.}
\label{FigTwoSided}
\end{center}
\end{figure}

The proof of Theorem~\ref{ThmTASEP} is in Section~\ref{subsectProofTASEP}. It is a consequence of the corresponding result for last passage percolation (see Theorem~\ref{ThmTwoSidedLPP}), together with the slow decorrelation phenomena (see Proposition~\ref{ThmSlowDec}). The stationary case, (a3), was analyzed in~\cite{BFP09} (see Theorem~1.7 therein); the $\tau_k^2$ term in (\ref{eq2.11}) compensates the fact that the scaling (\ref{eq2.8}) is not following the (straight) limit shape approximation. For presentation simplicity in~\cite{BFP09} only the fixed-time result was stated, but slow decorrelation allow immediately to extend it as in Theorem~\ref{ThmTASEP}.

\subsection{Directed percolation}\label{subsectLPP}
In Section~\ref{sectConnection} we explain the precise connection between TASEP and last passage percolation (LPP). In Proposition~\ref{ThmOneSidedLPP} we give the asymptotics for one-sided LPP, which uses the determinantal structure of the Schur process and slow decorrelation. The extension to two-sided boundary conditions via coupling arguments is stated in Theorem~\ref{ThmTwoSidedLPP}.

\subsubsection{Connection with TASEP}\label{sectConnection}
We define a directed last passage percolation model by assigning random waiting times $w_{i,j}$ to each site $(i,j)$ in $(\Z_+^*)^2$ ($\Z_+^*=\{0,1,\ldots\}$). We require $w_{i,j}$'s are independent and exponentially distributed variables\footnote{We use the notation ${\rm Exp}(m)$ for a random variable which is exponentially distributed with mean $m$.} (to be specified below). To every directed (up/right only) path $\pi$ from $(0,0)$ to $(x,y)$ we associate the waiting time $T(\pi)=\sum_{(i,j)\in \pi} w_{i,j}$. Then, the last passage time from $(0,0)$ to $(x,y)$ is the longest waiting time over all directed paths:
\begin{equation}
L(x,y) = \max_{\pi:(0,0)\to (x,y)} T(\pi).
\end{equation}

It is well known that the height function for TASEP with our initial condition is expressible in terms of a LPP model. Let us shortly recall this connection which is established in full generality in~\cite{PS01} and briefly reexplained~\cite{BFP09}, extending the step-initial condition case considered in~\cite{Jo00b}. We label particles from right to left and denote by $\textbf{x}_k(t)$ the position of particle $k$ at time $t$. We set the label so that $\cdots<\textbf{x}_{2}(0)<\textbf{x}_1(0)<0\leq \textbf{x}_0(0)<\textbf{x}_{-1}(0)<\cdots$.\\
(a) For $i,j\geq 1$, $w_{i,j}$ is the waiting time that particle $j$ jump from site $i-j-1$ to site $i-j$ (of course, the waiting time counted from the instant where site $i-j$ is empty). Thus, $w_{i,j}\sim {\rm Exp}(1)$ random variables, $i,j\geq 1$.\\
(b) On the other hand, the effect on the dynamics on $\textbf{x}_0(t)$ due to the particles on the its right is equivalent to set the jump rate of particle $0$ to be $1-\rho_+$ instead of $1$. This is a consequence of Burke's Theorem~\cite{Bur56}. Therefore we set $w_{i,0}\sim {\rm Exp}(1/(1-\rho_+))$ for $i> \textbf{x}_0(0)$ and $w_{i,1}=0$ otherwise.\\
(c) By looking at the particle-hole transformation, we set $w_{1,j}\sim {\rm Exp}(1/\rho_-)$ for $j\geq -\textbf{x}_1(0)$ and $w_{0,j}=0$ otherwise. Finally, we set $w_{0,0}=0$.

With this settings, the correspondence between last passage time, particle positions, and height function is the following: for $x_k,y_k\geq 1$, $t_k>0$, we have
\begin{equation}\label{eqTASEPvsDP}
\begin{aligned}
\Pb\left(\cap_{k=1}^m\{L(x_k,y_k)\leq t_k\}\right)
& =\Pb\left(\cap_{k=1}^m\{\textbf{x}_{y_k}(t_k)\geq x_k-y_k\}\right) \\
&= \Pb\left(\cap_{k=1}^m\{h_{t_k}(x_k-y_k)\geq x_k+y_k\}\right).
\end{aligned}
\end{equation}
Since $\textbf{x}_0(0)\sim {\rm Geom}(1-\rho_+)$ and $-(1+\textbf{x}_1(0))\sim {\rm Geom}(\rho_-)$, we can set $w_{i,1}\sim {\rm Exp}(1/(1-\rho_+))$ for all $i\geq 1$, and $w_{1,j}\sim {\rm Exp}(1/\rho_-)$ for all $j\geq 1$ without changing the large time asymptotics (see e.g.\ Proposition 2.2 in~\cite{FS05a}) (but keeping $w_{0,0}=0$).

\subsubsection{One-sided LPP}\label{subsectOneSided}
As briefly mentioned in the introduction, the proof of our result uses a mixture of analytic and probabilistic methods. On the analytic side, we have to analyze the following LPP model (referred as one-sided LPP):
\begin{equation}\label{eqWaitTimeLPPone}
w_{i,j} = \begin{cases}
\textrm{Exp}(1) &\textrm{for } i,j\geq 1,\\
\textrm{Exp}(1/\eta) &\textrm{for } i=0,j\geq 1,\\
0 &\textrm{for } i\geq 0,j=0,
\end{cases}
\end{equation}
where $\eta\in (0,1]$ is a constant. We denote by $L_1$ the last passage time for the waiting times (\ref{eqWaitTimeLPPone}), where $1$ stands for one-sided. This problem is related to TASEP with $\rho_+=0$ and $\rho_-=\eta$. Moreover, see Section~\ref{subsectRM}, the statistics of $L_1$ are related to those of the largest eigenvalue of a perturbed Wishart (sample covariance) matrix.

To set the scaling variables, we need an expression for the limit shape. Let us focus along the line $y=\gamma^2 x$. There are two cases:\\[0.5em]
(a) for $\gamma(1+\gamma)^{-1}\leq\eta\leq 1$,
\begin{equation}\label{eqLimShapeRM}
\lim_{T\to\infty} \frac{1}{T} L_1(\xi T,\xi \gamma^2 T)=\xi(1+\gamma)^2,
\end{equation}
(b) for $0<\eta\leq \gamma(1+\gamma)^{-1}$,
\begin{equation}\label{eqLimShapeGauss}
\lim_{T\to\infty} \frac{1}{T} L_1(T,\gamma^2 T)=\frac{1}{1-\eta}+\frac{\gamma^2}{\eta}.
\end{equation}

In the regime where the limit shape is (\ref{eqLimShapeRM}), the fluctuations of $L_1$ are of random matrix type with correlation on the $T^{2/3}$ scale and fluctuations in the $T^{1/3}$ scale. Therefore we introduce the scaling
\begin{equation}\label{scaling0}
\begin{aligned}
x(\tau)&=\left\lfloor \frac{1}{(1+\gamma)^2}T + \frac{2\tau}{(1+\gamma)^{2/3}\gamma^{2/3}}T^{2/3}\right\rfloor,\\
y(\tau)&=\left\lfloor \frac{\gamma^2}{(1+\gamma)^2} T\right\rfloor,\\
\ell(\tau,s)&= T+ \frac{2\tau(1+\gamma)^{1/3}}{\gamma^{2/3}} T^{2/3} +(s-\tau^2)\frac{(1+\gamma)^{2/3}}{\gamma^{1/3}}T^{1/3},
\end{aligned}
\end{equation}
where the parameter $s$ is a measure of the fluctuations with respect to $\ell(\tau,0)$, that is what we expect to see from (\ref{eqLimShapeRM}). Under this scaling, the height fluctuations are governed by the Airy$_2$ process, $\mathcal{A}_2$, up to the critical value $\eta=\gamma(1+\gamma)^{-1}$ where there is a transition process, $\mathcal{A}_{{\rm BM}\to 2}$, to the Brownian motion behavior. In the regime where the limit shape is (\ref{eqLimShapeGauss}), the fluctuations will be governed by the boundary sources. They have fluctuation on the $T^{1/2}$ scale, correlation length of order $T$ and limit process the Brownian Motion, $\mathcal{B}$. This is precisely stated in following theorem.

\begin{prop}\label{ThmOneSidedLPPbasic}
Let $\mathcal{A}_2$, $\mathcal{A}_{{\rm BM}\to 2}$ and $\mathcal{B}$ be the processes defined in Section~\ref{subsectDefProc}. \\
(a) Fix $m\in \N$, $\eta\in (0,1]$ and $\gamma\in (0,\infty)$ with $\eta\geq \gamma (1+\gamma)^{-1}$. Then, for any given $\tau_1<\tau_2<\ldots<\tau_m$ and $s_1,\ldots,s_m\in\R$, we have:\\
(a1) if $\eta> \gamma (1+\gamma)^{-1}$, then
\begin{equation}
\lim_{T\to\infty} \Pb\left(\bigcap_{k=1}^{m} \{L_1(x(\tau_k),y(\tau_k))\leq \ell(\tau_k,s_k)\}\right) =
\Pb\left(\bigcap_{k=1}^{m}\{\mathcal{A}_2(\tau_k)\leq s_k\}\right),
\end{equation}
(a2) while if $\eta=\gamma (1+\gamma)^{-1}$, then
\begin{equation}
 \lim_{T\to\infty} \Pb\left(\bigcap_{k=1}^{m} \{L_1(x(\tau_k),y(\tau_k))\leq \ell(\tau_k,s_k)\}\right) = \Pb\left(\bigcap_{k=1}^{m}\{\mathcal{A}_{{\rm BM}\to 2}(\tau_k)\leq s_k\}\right).
\end{equation}
(b) Fix $m\in \N$, $\eta\in (0,1]$. Then, for any given $\gamma_1<\gamma_2<\ldots<\gamma_m$ such that $\eta<\gamma_1(1+\gamma_1)^{-1}$, and $s_1,\ldots,s_m\in\R$, we have
\begin{multline}
\lim_{T\to\infty}\Pb\left(\bigcap_{k=1}^{m} \left\{L_1(T,\gamma_k^2 T)\leq \left(\frac{\gamma_k^2}{\eta}+\frac{1}{1-\eta}\right) T+ s_k T^{1/2}\right\}\right)\\
=\Pb\left(\bigcap_{k=1}^{m}\left\{\mathcal{B}\left(\left[\frac{\gamma_k^2}{\eta^2}-\frac{1}{(1-\eta)^2}\right]\right) \leq s_k\right\}\right).
\end{multline}
\end{prop}

This theorem is proved in Section~\ref{subsectProofOneSided} using the Schur process and applying methods of asymptotic analysis. Having established this theorem for the one-sided boundary condition model above we use coupling methods to prove a general two-sided boundary condition theorem.

To understand intuitively the cutoff $\eta=\gamma(1+\gamma)^{-1}$ there is a simple argument. The last passage path (the random directed path which achieves the last passage time) goes along the left boundary for some distance and then will depart. Restricting the set of paths to only those which go a certain macroscopic distance along the boundary and then depart into the bulk, one may use independence of the boundary and the bulk to establish a law of large number and fluctuation theorem for the restricted last passage time. If the mean $1/\eta$ of the boundary waiting times is large enough ($\eta$ small enough), the restricted law of large numbers will be maximized for a positive macroscopic distance along the boundary. In this case, the fluctuations will come entirely from the boundary fluctuations and they will be given by the standard CLT: Gaussian fluctuations on the scale $T^{1/2}$. On the other hand, if the boundary waiting times are too small, the restricted law of large numbers will be maximized for a distance along the boundary of $o(T^{2/3})$, and hence the fluctuations will come from the bulk, which are known to be $T^{1/3}$ and GUE Tracy-Widom distributed. At the cutoff, the two fluctuations compete and yield a perturbation of the bulk fluctuations. This intuition will be useful in some of the arguments used in this paper.

To extend Proposition~\ref{ThmOneSidedLPPbasic} to points which have different $y$ coordinates we use the fact that in certain directions (the characteristics) the last passage time fluctuations decorrelate not in the scale $T^{2/3}$ but rather in the scale of order $T$. This means that the passage time at two points at distance $o(T)$ but on the same characteristic have the same fluctuations (up to $o(T^{1/3})$). This phenomena, known as slow decorrelation\footnote{For flat interfaces as considered in~\cite{KK99} the dynamic scale invariance~\cite{KS92} implies a scaling form for the temporal autocorrelation, from which one could expect to see a slow-decorrelation type of phenomenon. This phenomenon is also related with what is known as \emph{persistence}, see~\cite{KK99} for KPZ class and~\cite{KKMCBS97} for Gaussian-type models.}, was observed (and proven) in the related PNG model in~\cite{Fer08} and then extended to a much greater generality within models in the KPZ universality class~\cite{CFP10b}. We recall the result needed here below.

The characteristic lines for TASEP with particle density $\rho$ move with speed $1-2\rho$. In the present LPP picture this implies the following:\\[0.5em]
(a) if $\gamma\leq\frac{\eta}{1-\eta}$, the line $y=\gamma^2 x$ is a characteristic related to the rarefaction fan of TASEP,\\[0.5em]
(b) while in the case $\gamma>\frac{\eta}{1-\eta}$, the characteristic passing by $(T,\gamma^2 T)$ is given by
\begin{equation}
y=T\gamma^2+(x-T)\eta^2(1-\eta)^{-2}.
\end{equation}
With these preliminaries, we can state the slow decorrelation theorem for LPP.

\begin{prop} [Corollary of Theorem 2.2 of~\cite{CFP10b}]\label{ThmSlowDec}
$ $\\
(a) For $\gamma\leq\frac{\eta}{1-\eta}$, define
\begin{equation}
P=\left(\left\lfloor \frac{T}{(1+\gamma)^2}\right\rfloor,\left\lfloor\frac{\gamma^2 T}{(1+\gamma)^2}\right\rfloor\right),\quad Q=\left(\left\lfloor \frac{T+r}{(1+\gamma)^2}\right\rfloor,\left\lfloor\frac{\gamma^2 (T+r)}{(1+\gamma)^2}\right\rfloor\right).
\end{equation}
Then, for any $r\sim T^{\nu}$ with $\nu\in [0,1)$ and any given $M>0$, it holds
\begin{equation}
\lim_{T\to\infty}\Pb\left(\left|L_1(Q)-L_1(P)-r\right|\geq M T^{1/3}\right)=0.
\end{equation}
(b) For $\gamma>\frac{\eta}{1-\eta}$, define
\begin{equation}P=(\lfloor T\rfloor,\lfloor\gamma^2 T\rfloor),\quad Q=\left(\lfloor T+r\rfloor,\left\lfloor\gamma^2 T+r\eta^2(1-\eta)^{-2}\right\rfloor\right).
\end{equation}
Then, for any $r\sim T^{\nu}$ with $\nu\in [0,3/2)$ and any given $M>0$, it holds
\begin{equation}
\lim_{T\to\infty} \Pb\left(\left|L_1(Q)-L_1(P)-\frac{r}{(1-\eta)^2}\right|\geq M T^{1/2}\right)=0.
\end{equation}
\end{prop}

\begin{remark}
In Proposition~\ref{ThmSlowDec} $\gamma$ can be also chosen to be $T$-dependent, provided that it converges to a fixed number in the $T\to\infty$ limit.
\end{remark}

\begin{remark}\label{decorrelation_remark}
We will use a generalization of Proposition~\ref{ThmSlowDec} to joint distributions. We let $r$ be a vector $r=(r_1,r_2,\ldots, r_m)$ with each \mbox{$r_i\sim T^{\nu_i}$}. Let us interprete $P$ and $L_1(P)$ as vectors $P=(P_1,P_2,\ldots, P_m)$ and \mbox{$L_1(P) = (L_1(P_1),L_1(P_2),\ldots, L_1(P_m))$}, and similarly for $Q$ and $L_1(Q)$. Then, the theorem still holds, with absolute values replaced by Euclidean norms. Case (a) holds for $\max(\nu_i)<1$ and case (b) for $\max(\nu_i)<3/2)$. This follows directly from triangular inequality the union probability bound.
\end{remark}

As immediate application of Proposition~\ref{ThmSlowDec} (and Remark~\ref{decorrelation_remark}) is the extension of Proposition~\ref{ThmOneSidedLPPbasic} away from the fixed-$y$ line. Indeed, often one considers the cut $x+y=t$ and $t$ is then interpreted as the time parameter in a stochastic growth model (see e.g.~\cite{Jo03b,PS02}). For that reason we consider the following modification of the scaling (\ref{scaling0}): for a given $\nu\in [0,1)$,
\begin{equation}\label{scaling1}
\begin{aligned}
x(\tau,\theta)&=\left\lfloor \frac{1}{(1+\gamma)^2}(T + \theta T^\nu) +\tau \frac{2\gamma^{4/3}}{(1+\gamma^2)(1+\gamma)^{2/3}}T^{2/3}\right\rfloor,\\
y(\tau,\theta)&=\left\lfloor \frac{\gamma^2}{(1+\gamma)^2} (T+\theta T^\nu)-\tau \frac{2\gamma^{4/3}}{(1+\gamma^2)(1+\gamma)^{2/3}}T^{2/3}\right\rfloor,\\
\ell(\tau,\theta,s)&= T+\theta T^\nu+ \tau\frac{2\gamma^{1/3}(1+\gamma)^{1/3}(\gamma-1)}{1+\gamma^2} T^{2/3} + (s-\tau^2)\frac{(1+\gamma)^{2/3}}{\gamma^{1/3}}T^{1/3}.
\end{aligned}
\end{equation}
One might have noticed that in the scaling (\ref{scaling1}) we extend $x$ and $y$ along the characteristic for $\tau=0$, which is not exactly the characteristic for $\tau\neq 0$. However, this is not a problem, since the projection of $(x(\tau,\theta),y(\tau,\theta))$ along the true characteristic on the line $x+y=\frac{1+\gamma^2}{(1+\gamma)^2}T$ is $(x(\tilde \tau,0),y(\tilde \tau,0))$ with $\tilde \tau=\tau+o(1)$. Then, the extension of Proposition~\ref{ThmOneSidedLPPbasic} to the scaling (\ref{scaling1}) is the following.

\begin{prop}\label{ThmOneSidedLPP}
$ $\\
(a) Fix $m\in \N$, $\nu\in [0,1)$, $\eta\in (0,1]$, and $\gamma\in (0,\infty)$ such that \mbox{$\eta\geq \gamma (1+\gamma)^{-1}$}. Then, for any given real numbers $\tau_1<\tau_2<\ldots<\tau_m$, $\theta_1,\ldots, \theta_m$ and $s_1,\ldots, s_m$, we have:\\
(a1) if $\eta > \gamma (1+\gamma)^{-1}$, then
\begin{equation}
\lim_{T\to\infty} \Pb\left(\bigcap_{k=1}^{m} \{L_1(x(\tau_k,\theta_k),y(\tau_k,\theta_k))\leq \ell(\tau_k,\theta_k,s_k)\}\right) = \Pb\left(\bigcap_{k=1}^{m}\{\mathcal{A}_2(\tau_k)\leq s_k\}\right),
\end{equation}
(a2) while if $\eta=\gamma(1+\gamma)^{-1}$, then
\begin{equation}
\lim_{T\to\infty} \Pb\left(\bigcap_{k=1}^{m} \{L_1(x(\tau_k,\theta_k),y(\tau_k,\theta_k))\leq \ell(\tau_k,\theta_k,s_k)\}\right) = \Pb\left(\bigcap_{k=1}^{m}\{\mathcal{A}_{{\rm BM}\to 2}(\tau_i)\leq s_k\}\right).
\end{equation}
(b) Fix $m\in \N$, $\eta\in(0,1]$. Then, for any given $\gamma_1<\gamma_2<\ldots<\gamma_m$ such that $\eta<\gamma_1(1+\gamma_1)^{-1}$, and $s_1,\ldots, s_m\in\R$, it holds
\begin{multline}\label{eq2.15}
\lim_{T\to\infty}\Pb\left(\bigcap_{k=1}^{m} \left\{L_1(\theta_k T,\gamma_k^2\theta_k T)\leq \left(\frac{\gamma_k^2}{\eta}+\frac{1}{1-\eta}\right)\theta_k T+ s_k T^{1/2}\right\}\right)\\ =\Pb\left(\bigcap_{k=1}^{m}\left\{\mathcal{B}\left(\theta_k\left[\frac{\gamma_k^2}{\eta^2}-\frac{1}{(1-\eta)^2}\right]\right) \leq s_k\right\}\right).
\end{multline}
\end{prop}

\subsubsection{Two-sided LPP}
The main object of interest in this paper is last passage percolation with two-sided boundary conditions defined as follow. Given two paramaters \mbox{$\pi,\eta\in (0,1]$}, the independent waiting times $w_{i,j}$ satisfy
\begin{equation}\label{LPP_weights}
w_{i,j} = \begin{cases}
       \textrm{Exp}(1/\pi) &\textrm{for } i\geq 1,j=0,\\
           \textrm{Exp}(1/\eta) &\textrm{if } i=0,j\geq 1,\\
       \textrm{Exp}(1) &\textrm{if } i,j\geq 1,\\
           0& \textrm{for } i=0,j=0.\\
          \end{cases}
\end{equation}
We denote by $L_2$ the last passage percolation time for the waiting times (\ref{LPP_weights}), where the subscript $2$ stands for the two-sided.
This corresponds with TASEP with two-sided Bernoulli initial conditions. The connection with two-sided Bernoulli initial condition for TASEP is obtained by setting $\eta=\rho_-$ and $\pi=1-\rho_+$.

The new phenomenon that occurs for two-sided LPP with respect to one-sided is the possible presence of shockwaves in the corresponding TASEP picture. This occurs when $\pi+\eta<1$, i.e., when characteristics meet. Indeed, the characteristic leaving from the axis $(\R_+,0)$ have slope $(1-\pi)^2\pi^{-2}$ and, whenever $\eta+\pi<1$, they meet the characteristics leaving from the $(0,\R_+)$ axis, whose slope is $\eta^2(1-\eta)^{-2}$. The slope of the shockwave is determined by the Rankine-Hugoniot condition and it is given by the equation
\begin{equation}\label{eqShockWave}
y=\frac{\eta(1-\pi)}{\pi(1-\eta)}x.
\end{equation}
The limit shape is not anymore always as in (\ref{eq2.15}) but depends on the which side of the shockwave we focus on $S(\gamma)\equiv\lim_{T\to\infty}\frac{1}{T}L_2(T,\gamma^2 T)$ given by
\begin{equation}
S(\gamma)= \begin{cases}
{\displaystyle \frac{\gamma^2}{\eta}+\frac1{1-\eta}},&\textrm{if }\pi\leq (1+\gamma)^{-1}, \eta<\gamma(1+\gamma)^{-1}, \gamma^2>\frac{\eta(1-\pi)}{\pi(1-\eta)},\\[0.5em]
{\displaystyle \frac{\gamma^2}{1-\pi}+\frac1{\pi}},&\textrm{if }\pi< (1+\gamma)^{-1}, \eta\leq\gamma(1+\gamma)^{-1}, \gamma^2<\frac{\eta(1-\pi)}{\pi(1-\eta)}.
\end{cases}
\end{equation}
When $\eta+\pi<1$, then the fluctuations are dominated by the boundary terms and live on a $T^{1/2}$ scale, while the bulk contribution to the fluctuations is only on a $T^{1/3}$ scale. Therefore, the limit process describing the fluctuations on each side of (not including) the shockwave is given by the Brownian motion obtained as the boundary contribution from the origin to the projections along the characteristics of the points we focus on. Thus, the two side of the shockwave will be independent. On the shockwave, there is a balance between the two boundary contributions: the last passage time for a point $P$ on the shockwave is the maximum between the last passage time of the one-sided problem with $w_{i,0}=0$ and the transposed one-sided problem $w_{0,j}=0$. Since the fluctuations come only from the boundaries, the distribution of $P$ will be the product of the distribution of the two one-sided problems, see Figure~\ref{FigTwoSided}.

This intuitive picture is confirmed by Theorem~\ref{ThmTwoSidedLPP}, which can be obtained from Proposition~\ref{ThmOneSidedLPP} without any additional hard analysis by using coupling arguments introduced in~\cite{BC09}.

\begin{thm}\label{ThmTwoSidedLPP}
Consider the same scaling (\ref{scaling1}) as in Proposition~\ref{ThmOneSidedLPP}.\\
(a) Fix $m\in \N$, $\eta,\pi\in (0,1]$, $\gamma\in (0,\infty)$, and $\nu\in [0,1)$. Then, for any choice of real numbers $\tau_1<\tau_2<\ldots,\tau_m$, $\theta_1,\ldots,\theta_m$, and $s_1,\ldots,s_m$, we have:\\
(a1) If $\pi>(1+\gamma)^{-1}$ and $\eta>\gamma(1+\gamma)^{-1}$, then
\begin{equation}
\lim_{T\to\infty} \Pb\bigg(\bigcap_{k=1}^{m} \{L_2(x(\tau_k,\theta_k),y(\tau_k,\theta_k))\leq \ell(\tau_k,\theta_k,s_k)\}\bigg) = \Pb\bigg(\bigcap_{k=1}^{m}\{\mathcal{A}_2(\tau_k)\leq s_k\}\bigg).
\end{equation}
(a2) If $\pi>(1+\gamma)^{-1}$ and $\eta=\gamma(1+\gamma)^{-1}$, then
\begin{equation}\label{eqL21}
\lim_{T\to\infty} \Pb\bigg(\bigcap_{k=1}^{m} \{L_2(x(\tau_k,\theta_k),y(\tau_k,\theta_k))\leq \ell(\tau_k,\theta_k,s_k)\}\bigg) = \Pb\bigg(\bigcap_{k=1}^{m}\{\mathcal{A}_{{\rm BM}\to 2}(\tau_k)\leq s_k\}\bigg).
\end{equation}
(a3) If $\pi=(1+\gamma)^{-1}$ and $\eta=\gamma(1+\gamma)^{-1}$, then
\begin{equation}
\lim_{T\to\infty} \Pb\bigg(\bigcap_{k=1}^{m} \{L_2(x(\tau_k,\theta_k),y(\tau_k,\theta_k))\leq \ell(\tau_k,\theta_k,s_k)\}\bigg) = \Pb\bigg(\bigcap_{k=1}^{m}\{\mathcal{A}_{\rm stat}(\tau_k)\leq s_k+\tau_k^2\}\bigg).
\end{equation}

(b) Fix $m_l,m_s,m_b\in \Z_+^*$ and set $m=m_l+m_s+m_b$. $m_l$ is the number of points associated with characteristics from the left boundary of the LPP, $m_b$ to the bottom boundary, and $m_s$ on the shockwave (if exists). For $\pi,\eta\in (0,1]$ and real numbers $\theta_1,\ldots, \theta_m$, choose $\gamma_i$ corresponding to each case. Then
\begin{equation}\label{bmeqn}
\begin{aligned}
\lim_{T\to \infty} &\Pb\bigg(\bigcap_{k=1}^{m} \{L_2(\theta_k T,\gamma_k^2\theta_k T)\leq S(\gamma_k)\theta_k T+s_k T^{1/2}\}\bigg)\\
=& \Pb\bigg(\bigcap_{k=1}^{m_l+m_s}\left\{\mathcal{B}\left(\theta_k\left[\frac{\gamma_k^2}{\eta^2}-\frac{1}{(1-\eta)^2}\right]\bigg)\leq s_k\right\}\right) \\
\times &\Pb\bigg(\bigcap_{k=m_l+1}^m\left\{\mathcal{B}'\left(\theta_k\left[\frac{1}{\pi^2}-\frac{\gamma_k^2}{(1-\pi)^2}\right]\right)\leq s_k\right\}\bigg),
\end{aligned}
\end{equation}
where $\mathcal{B}$ and $\mathcal{B}'$ are two independent copies of Brownian Motion.
\end{thm}

\begin{remark} The case $\pi=(1+\gamma)^{-1}$ and $\eta>\gamma(1+\gamma)^{-1}$ can be recovered by Theorem~\ref{ThmTwoSidedLPP} (a2) by the change of variable $(x,y)\to (y,x)$ and $\gamma\to \gamma^{-1}$.
\end{remark}

Because of the nice correspondence with the TASEP, we focused here on LPP with two-sided ``of width one'', i.e., with modified weights only for $i=0$ and $j=0$ in (\ref{LPP_weights}). From a LPP point of view a natural extension is to consider weights different from $1$ for a larger number of columns/rows. For example, the one-sided with boundary width equal to $r$ was considered in~\cite{BBP06}. Like for the Airy$_2$ process~\cite{AvM03}, also one can describe joint distributions by PDE's~\cite{AvMD08}. For an extension to two-sided thick boundaries, see~\cite{BP07}. The particular case of boundaries of sizes $1$ and $r$ corresponds, in terms of TASEP, to Bernoulli initial condition on $\Z_-$ with the first $r$ particles having a different jump rate~\cite{BFS09}. The coupling techniques would also work, except to the critical cases (like $\eta+\pi=1$).

\subsection{Random sample covariance matrices}\label{subsectRM}
We now give a random matrix interpretation of our result. For that purpose, we use a result proven in~\cite{DW08}.
Consider an infinite array \mbox{$A(N)=[A_{i,j}]_{i\geq 1,1\leq j \leq N}$} where the $A_{i,j}$ are independent complex Gaussian random variables with mean zero and variance $1/(\pi_i+\tilde \pi_j)$.
Define $A(n,N)$ to be the matrix obtained by considering the first $n$ rows from $A$. Define the $N\times N$ matrix $M_N(n):=A(n,N)^*A(n,N)$ and denote by $\lambda_{N,\max}(n)$ its largest eigenvalue.

Let also $L(n,N)$ be the last passage time from $(1,1)$ to $(n,N)$ in the percolation model with independent exponential random variables with
expectations $1/(\pi_i +\tilde \pi_j)$. Then it is proved in~\cite{BP07} and~\cite{DW08} that the process $n\mapsto \lambda_{N,\max}(n)$ and the process of last passage times $n\mapsto L(n,N)$ have the same distributions. Therefore, from Proposition~\ref{ThmOneSidedLPPbasic} we deduce the following.

\begin{thm}\label{ThmRM}
Set $\pi_i=1-\eta, i\geq 2$, $\pi_1=0$, $\tilde \pi_j=\eta, j\geq 1$.\\
(a)  Fix $m\in \N$, $\gamma\in (0,\infty)$ and $\eta\in (0,1]$. Recall $x(\tau_k)$, $y(\tau_k)$, and $\ell(\tau_k,s_k)$ defined in~(\ref{scaling0}). Set $N=y(\tau_k)=\left\lfloor\frac{\gamma^2}{(1+\gamma)^2}T\right\rfloor$.\\
(a1) If $\eta>\frac{\gamma}{1+\gamma}$, then for real numbers $\tau_1<\tau_2<\ldots<\tau_m$ and $s_1,\ldots, s_m$ it holds
\begin{equation}
\lim_{T\to\infty} \Pb\left(\bigcap_{k=1}^{m} \{\lambda_{N,\max}(x(\tau_k))\leq \ell(\tau_k,s_k)\}\right) = \Pb\left(\bigcap_{i=1}^{m}\{{\mathcal{A}_2}(\tau_i)\leq  s_i\}\right).
\end{equation}
(a2) If $\eta=\frac{\gamma}{1+\gamma}$, then for real numbers $\tau_1<\tau_2<\ldots<\tau_m$ and $s_1,\ldots, s_m$ it holds
\begin{equation}
 \lim_{T\to\infty} \Pb\left(\bigcap_{k=1}^{m} \{\lambda_{N,\max}(x(\tau_k))\leq \ell(\tau_k,s_k)\}\right) = \Pb\left(\bigcap_{i=1}^{m}\{\mathcal{A}_{{\rm BM}\to 2}(\tau_i)\leq  s_i\}\right).
\end{equation}
(b) Fix $m\in \N$, $\eta\in(0,1]$. Set $N=\lfloor T\rfloor$, then for any $\gamma_1<\gamma_2<\ldots<\gamma_m$ such that $\eta<\gamma_1(1+\gamma_1)^{-1}$, and $s_1,\ldots, s_m\in\R$, it holds
\begin{multline}
\lim_{T\to\infty}\Pb\left(\bigcap_{k=1}^{m} \left\{\lambda_{N,\max}(\gamma_k^2 T)\leq \left(\frac{\gamma_k^2}{\eta}+\frac{1}{1-\eta}\right) T+ s_kT^{1/2}\right\}\right)\\ =\Pb\left(\bigcap_{k=1}^{m}\left\{\mathcal{B}\left(\left[\frac{\gamma_k^2}{\eta^2}-\frac{1}{(1-\eta)^2}\right]\right) \leq s_k\right\}\right).
\end{multline}
\end{thm}

\subsection{Limit processes: definitions}\label{subsectDefProc}
Here we collect the definitions of the limit processes.

\begin{defin}[Airy$_2$ process, $\mathcal{A}_2$]\label{DefAiry2}
The Airy$_2$ process is defined in terms of finite dimensional distributions as
\begin{equation}
\Pb\left(\bigcap_{k=1}^{m}\{\mathcal{A}_2(\tau_k)\leq s_k\}\right)=\det(\Id-\chi_s K_{\mathcal{A}_2} \chi_s)_{L^2(\{\tau_1,\ldots, \tau_m\}\times \R)},
\end{equation}
where $\chi_s(\tau_k,x) = \Id_{[x>s_k]}$, and $K_{\mathcal{A}_2}$ is the extended Airy kernel:
\begin{equation}
K_{\mathcal{A}_2}(\tau,s;\tau',s')=\begin{cases}
{\displaystyle \int_{\R_+} dz e^{(\tau'-\tau)z}\Ai(s+z)\Ai(s'+z)}, & \tau\geq \tau',\\
{\displaystyle -\int_{\R_-} dz e^{(\tau'-\tau)z}\Ai(s+z)\Ai(s'+z)}, & \tau<\tau'.
\end{cases}
\end{equation}
\end{defin}
The Airy$_2$ process was discovered in the PNG model~\cite{PS02}. It is a stationary process with one-point distribution given by the GUE Tracy-Widom distribution $F_2$~\cite{TW94}. An integral representation of $K_{\mathcal{A}_2}$ can be found in Proposition~2.3 of~\cite{Jo03b}; another form is in Definition 21 of~\cite{BFS09} in the $M=0$ case.

\begin{defin}\label{DefBMtoAiry2}
We denote by $\mathcal{A}_{{\rm BM}\to 2}$ the transition process from Brownian Motion to Airy$_2$. It is defined in terms of finite dimensional distributions as
\begin{equation}
\Pb\left(\bigcap_{k=1}^{m}\{\mathcal{A}_{{\rm BM}\to 2}(\tau_k)\leq s_k\}\right)=\det(\Id-\chi_s K_{\mathcal{A}_{{\rm BM}\to 2}} \chi_s)_{L^2(\{\tau_1,\ldots, \tau_m\}\times \R)},
\end{equation}
where $\chi_s(\tau_k,x) = \Id_{[x>s_k]}$, and $K_{\mathcal{A}_{{\rm BM}\to 2}}$ is the rank-one perturbation $K_{\mathcal{A}_2}$:
\begin{equation}
\begin{aligned}
K_{\mathcal{A}_{{\rm BM}\to 2}}(\tau,s;\tau',s')=K_{\mathcal{A}_2}(\tau,s;\tau',s')+\Ai(s)\left(e^{\frac13\tau'^3-s'\tau'}-\int_{\R_+} dz e^{\tau'z}\Ai(s'+z)\right).
\end{aligned}
\end{equation}
\end{defin}

This transition process was derived in~\cite{SI04}. In Equation~(3.6) of~\cite{SI04} the kernel is divided into two cases. However, using the identity~(D.3) in~\cite{FS05a}, namely, $\int_\R dy e^{wy}\Ai(\beta+y)=e^{w^3/3-\beta w}$, we can rewrite
\begin{equation}
K_{\mathcal{A}_{{\rm BM}\to 2}}(\tau,s;\tau',s')=K_{\mathcal{A}_2}(\tau,s;\tau',s')+\Ai(s)\int_{\R_+}dz e^{-\tau' z} \Ai(s'-z)
\end{equation}
for $\tau'>0$.

An integral representation of the kernel $K_{\mathcal{A}_{{\rm BM}\to 2}}$ can be found in~\cite{BFS09}, Definition 21, in the $M=1$ case. To see the Brownian Motion behavior, one has to take the $\tau\ll -1$ and replace $s$ by $s+\tau^2$. This shift is needed to take into account that the actual limit shape at the critical point changes from (\ref{eqLimShapeRM}) to (\ref{eqLimShapeGauss}). So, for large $-\tau$, the approximation coming from (\ref{eqLimShapeRM}) is not optimal anymore. Indeed, using (\ref{eqLimShapeGauss}), $s-\tau^2$ would be replaced by $s$ in (\ref{scaling0}).

The definition of the process for stationary TASEP, $\mathcal{A}_{\rm stat}$, is quite intricate~\cite{BFP09}. Its joint distributions is the r.h.s.\ of Equation (1.9) in~\cite{BFP09}.

\begin{defin}
The last process, $\mathcal{B}$, is simply a standard one dimensional Brownian motion. Its finite dimensional distributions can be expressed in terms of a Fredholm determinant: let $0<\tau_1<\cdots < \tau_m$, then
\begin{equation}
\Pb\left(\bigcap_{k=1}^{m}\{\mathcal{B}(\tau_k)\leq s_k\}\right)=\det(\Id-\chi_s K_{\cal B} \chi_s)_{L^2(\{\tau_1,\ldots, \tau_m\}\times \R)},
\end{equation}
where $\chi_s(\tau_k,x) = \Id_{[x>s_k]}$, and the kernel $K_{\cal B}$ is given by
\begin{equation}
K_{\cal B}(\tau,s;\tau',s') =
\frac{1}{\sqrt{2\pi \tau}}\exp\left(-\frac{s^2}{2\tau}\right) - \frac{\Id_{[\tau>\tau']}}{\sqrt{2\pi (\tau-\tau')}}\exp\left(-\frac{(s-s')^2}{2(\tau-\tau')}\right).
\end{equation}
\end{defin}

\section{Proof of results}\label{sectproofs}

\subsection{Proof of Proposition~\ref{ThmOneSidedLPPbasic}}\label{subsectProofOneSided}
Let $n,N$ be positive integers.
Consider the directed percolation model with one-sided boundary conditions (\ref{eqWaitTimeLPPone}).
Let $L_1(n,N)$ be the last passage times from $(0,0)$ to $(n,N)$.
The joint distribution of $L_1(n,N), n\geq 0$ can be analyzed thanks to the so-called Schur process studied in~\cite{BP07}\footnote{The $(N,p)$ in~\cite{BP07} corresponds to $(n-1,N)$ in this paper.}.
In particular the joint distribution of the last passage times in the directed percolation model is given by:
\begin{equation}\label{cor_kernel}
\Pb\left(\bigcap_{k=1}^m\{L_1(n_k,N)\leq S_k\}\right)=\det(\Id-P_S K_N P_S)_{L^2(\{n_1,\ldots,n_m\}\times\R)},
\end{equation}
where $P_S(k,x):=\Id_{[x>S_k]}$ and $K_N$ is the correlation kernel given by
\begin{equation}\label{K_N1}
\begin{aligned}
K_N(n_i,x;n_j,y)&=-\Psi_{n_i,n_j}(x,y)+ K_N^1(n_i,x;n_j,y),\\
K_N^1(n_i,x;n_j,y)&=\frac{1}{(2\pi\I)^2}\oint_{\cal C}dz\oint_{\cal C'}dw \frac{e^{wy-zx}}{w-z}\frac{(z+1-\eta)^{n_i-1}}{(w+1-\eta)^{n_j-1}}\frac{z}{w}\frac{(w-\eta)^N}{(z-\eta)^N},\\
\Psi_{n_i,n_j}(x,y)&=\Id_{[n_i<n_j]}\Id_{[x<y]}\frac{1}{2\pi\I}\oint_{\cal C'}dw e^{w(y-x)}(w+1-\eta)^{n_i-n_j}.
\end{aligned}
\end{equation}
where $\cal C$ (resp.\  $\cal C'$) is a contour oriented anticlockwise and enclosing $\eta$ (resp.\  $0$ and $\eta-1$).

\subsubsection{The case where $\eta> \frac{\gamma}{1+\gamma}$}
Consider the asymptotics of the correlation kernel with the rescaling\footnote{Comparing (\ref{rescaling}) with (\ref{scaling0}) one sees two minor differences: (a) the integer parts are not explicitly written, but it is irrelevant for the large $T$ asymptotics and (b) $\tilde\ell$ does not have the shift by $-\tau_i^2$ on $s_i$. This is not a problem, because in our case the $\tau_i$ are chosen in a bounded set. The scaling (\ref{scaling0}) is reobtained in the end by replacing $s_i+\tau_i^2$ by $s_i$.}
\begin{equation}\label{rescaling}
\begin{aligned}
n_i&=\left ( \frac{1}{1+\gamma}\right)^2T+\frac{2\tau_i}{(1+\gamma)^{2/3}\gamma^{2/3}}T^{2/3},\quad
N=\left ( \frac{\gamma}{1+\gamma}\right)^2T,\\
x_i&=\tilde \ell(\tau_i, s_i)=T+\frac{2\tau_i(1+\gamma)^{1/3}}{\gamma^{2/3}}T^{2/3}+s_i \frac{(1+\gamma)^{2/3}}{\gamma^{1/3}}T^{1/3}.
\end{aligned}
\end{equation}
We fix some $s_0\in \R$ and assume that $s_i\geq s_0$ for any $i$. To give the result we need a few definitions. Let us set
\begin{equation}\label{def wc}
 w_c=\eta-\frac{\gamma}{1+\gamma},\quad \chi= \frac{(1+\gamma)^{1/3}}{\gamma^{2/3}},\quad \rho= \frac{(1+\gamma)^{2/3}}{\gamma^{1/3}}=\frac{1+\gamma}{\gamma \chi}
\end{equation}
and $Z(i):= \exp\left(2\tau_i^3/3+\tau_i s_i+TF_i(w_c)+s_iT^{1/3}\rho w_c\right)$ where
\begin{equation}
F_i(w)= w\left(1+ \frac{2\tau_i\chi}{ T^{1/3}}\right)   +\frac{\gamma^2}{(1+\gamma)^2} \ln (w-\eta)-\left (\frac{1}{(1+ \gamma)^2}+\frac{2\tau_i}{ T^{1/3} \gamma \rho}\right ) \ln (w+1-\eta).\label{def: Fi}
\end{equation}
\begin{prop} \label{prop: asym+decayKN}
Uniformly for $s_i,s_j$ in a bounded interval, it holds
\begin{equation}
\begin{aligned}
&\lim_{N \to \infty} \rho T^{1/3}\frac{Z(i)}{Z(j)}K_N^1(n_i,x_i; n_j,x_j)\\
&=\int_{0}^{\infty}e^{-\lambda (\tau_i-\tau_j)}\Ai(s_i+\tau_i^2+\lambda)\Ai(s_j+\tau_j^2+\lambda)d\lambda +\Or(T^{-1/3}).
\end{aligned}
\end{equation}
Furthemore, for any $\kappa>0$, there exists a $T_0$ large enough such that
\begin{equation}
\Big|\rho T^{1/3}\frac{Z(i)}{Z(j)}K_N^1(n_i,x_i; n_j,x_j)\Big|\leq C e^{-\kappa (s_i+s_j)}
\end{equation}
for all $s_i,s_j\in\R$ and $T\geq T_0$. The constant $C$ is uniform in $T\geq T_0$ and $s_i,s_j$.
\end{prop}

\begin{proof}[Proof of Proposition~\ref{prop: asym+decayKN}]
The proof of Proposition~\ref{prop: asym+decayKN} relies on a saddle point analysis of the correlation kernel (\ref{K_N1}) with the rescaling (\ref{rescaling}).
We first rewrite the singularity $1/(w-z)$ in the kernel (\ref{K_N1}) as
\begin{equation}
\frac{1}{w-z}=-\int_0^{\infty} e^{\lambda(w-z) \rho T^{1/3}} \rho T^{1/3}d\lambda.
\end{equation}
This allows us to rewrite $K_N^1$ as a product of two integral kernels:
\begin{equation}
\rho T^{1/3}K_N^1(n_i,x_i; n_j,x_j)= \int_0^{\infty}H(\tau_i, s_i+\lambda)G(\tau_j, s_j+\lambda) d\lambda,
\end{equation}
where
\begin{equation}
\begin{aligned}
 H(\tau_i,s_i)&=\frac{\rho T^{1/3}}{2\pi\I} \oint_{\widetilde {\cal C}} e^{-Tz -2\tau_i \chi T^{2/3}z} z
\frac{(z+1-\eta)^{\frac{T}{(1+\gamma)^2} +2\tau_i \frac{T^{2/3}}{\rho \gamma}}}{(z-\eta)^{T(\frac{\gamma}{1+\gamma})^2}} e^{-s_i T^{1/3}\rho z}dz, \\
G(\tau_i,s_i)&=\frac{\rho T^{1/3}}{2\pi\I}\oint_{\cal C'} e^{Tw +2\tau_i \chi T^{2/3}w} \frac1w
\frac{(w-\eta)^{T(\frac{\gamma}{1+\gamma})^2}}{(w+1-\eta)^{\frac{T}{(1+\gamma)^2} +2\tau_i \frac{T^{2/3}}{\rho \gamma}}} e^{s_i T^{1/3}\rho w }dw.
\end{aligned}
\end{equation}
The contour $\widetilde {\cal C}$ is like $\cal C$ but instead of anticlockwise it is clockwise oriented. Note that the two contours $\widetilde {\cal C}$ and $\cal C'$ still cannot cross each other. To perform a saddle point analysis of both $G$ and $H$, we consider the first order leading terms in the exponential. Recall $F_i$ from (\ref{def: Fi}). Then for $\tau_i$ in a uniformly (in $T$) bounded interval, $F_i$ admits $w_c$ defined in (\ref{def wc}) as unique critical point. Furthermore \begin{equation}
F_i''(w_c)=2(1+\gamma)^2\frac{\tau_i}{\gamma \rho} T^{-1/3}, \quad F^{(3)}_i(w_c)=-2 \frac{(\gamma+1)^2}{\gamma}.
\end{equation}

Now we briefly expose the ideas of the asymptotics since such arguments have already been developed many times (see e.g.~\cite{BBP06} or at the beginning of the proof of Lemma~6.1 in~\cite{BF08}, where the steps are explained). From the assumption $\eta >\frac{\gamma}{1+\gamma}$ it follows $w_c>0$.
Consider the contours
\begin{equation}
{\cal C}_1=\{w_c+ te^{\pi\I/3}, t\in \R\}, \quad {\cal C}'_1=\{w_c+ te^{2\pi\I/3},0\leq t \leq 2\}.
\end{equation}
Then one has that
\begin{equation}
 \frac{d}{dt} \Re F_i\left(w_c+\frac{t}{\gamma+1}e^{2\pi\I/3}\right)= \frac{ -t^4+(\gamma-1)t^3-2\gamma t^2}{(t^2-t+1)(t^2+\gamma t +\gamma^2)}
-\tau_i \chi \frac{t(t+1)}{1-t+t^2}T^{-1/3}.\label{decF}
\end{equation}
The denominators in (\ref{decF}) is positive, being (product of) squared distances between the poles of the integrand and the integration path, e.g., $t^2-t+1= (\gamma+1)^2|w_c+ \frac{t}{\gamma+1}e^{2\pi\I/3}+1-\eta|^2.$
In particular $\Re F_i$ decreases along ${\cal C}'_1$. We now complete the contour in the upper half-plane as follows.
Call $w_0$ the endpoint of ${\cal C}'_1$ and set $r:= |w_0-\eta+1|$. Let $0<\theta_0<\pi$ be such that $w_0=\eta-1+re^{\I \theta_0}.$
Define
\begin{equation}
{\cal C}'_2:=\{\eta-1+r e^{\I\theta},\theta_0\leq \theta \leq \pi\}.
\end{equation}
Then it is not hard to see that $\Re F_i$ decreases along ${\cal C}'_2$:
\begin{equation}
\frac{d}{d\theta }\Re F_i\left(\eta-1+r e^{\I\theta}\right)=-r\sin \theta \left (1+\frac{2\tau_i\chi}{T^{1/3}}-\frac{1}{|-1+re^{\I\theta}|^2}\right)<-cr\sin \theta ,
\end{equation}
for some constant $c>0$. Thus ${\cal C}'={\cal C}'_1\cup {\cal C}'_2 \cup \overline{{\cal C}'_1\cup {\cal C}'_2}$ is a steep descent path for $F_i$.

For the $z-$contour, one has that
\begin{equation}
\frac{d}{dt}\Re F_i\left(w_c+ \frac{t}{\gamma+1}e^{\pi\I/3}\right)= \frac{t^4+t^3(\gamma-1)+2\gamma t^2}{(1+t+t^2)(t^2-\gamma t +\gamma^2)} +\tau_i \chi \frac{t(t-1)}{1+t+t^2} T^{-1/3}.
\end{equation}
In particular $\frac{d}{dt}\Re  F_i(w_c+ \frac{t}{\gamma+1}e^{\pi\I/3})>0$ as soon as $t\geq \frac12\gamma \tau_i \chi T^{-1/3}.$
From the latter we deduce that the main contribution to the $z-$integral will come from a $T^{-1/3}$ neighborhood of $w_c$. The contour ${\cal C}_1$ is not a steep ascent contour for $F_i$ but is enough for the purpose of evaluating the integral: it is a steep ascent path for the first order approximation of $F_i$, that is forgetting for a while the $\Or(T^{-1/3})$ terms in $F_i$.

In order to take care of the constraint on the contours $\widetilde {\cal C}$ and $\cal C'$ which cannot cross or touch each other, we now deform the $w$ and $z$ contours in a $T^{-1/3}$ neighborhood of $w_c$ so that the $w$- (resp.\ $z$-) contour lies to the left (resp.\  right) of $w_c$ (see Figure~\ref{FigContours}). As $F_i''(w)=\Or(T^{-1/3}),$ and \mbox{$F_i^{(3)} (w)=\Or(1)$} in such a neighborhood, the fact that the two contours are moved of $\Or(T^{-1/3})$ from the critical point has no impact on the asymptotics: this follows from a straightforward Taylor expansion of the exponential term.
\begin{figure}
\begin{center}
\psfrag{C}[c]{$\widetilde{\cal C}$}
\psfrag{C'}[c]{${\cal C}'$}
\psfrag{wc}[c]{$w_c$}
\psfrag{e}[c]{$\Or(T^{-1/3})$}
\includegraphics[height=4cm]{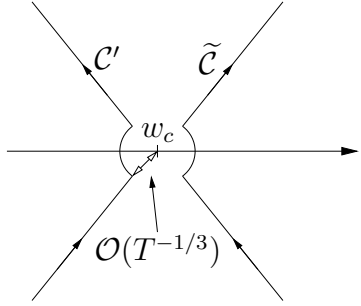}
\caption{The contours $\widetilde{\cal C}$ and $\cal C'$ are slightly deformed in the neighborhood of $w_c$ so that they don't touch each other.}
\label{FigContours}
\end{center}
\end{figure}

We then make the change of variables
\begin{equation}
w=w_c+\frac{s}{\rho T^{1/3}}, \quad z=w_c+\frac{t}{\rho T^{1/3}}.
\end{equation}
Then it is not hard to see that for bounded $s_i,s_j$ it holds
\begin{equation}\label{asyKN}
\begin{aligned}
\frac{1}{w_c}e^{TF_i(w_c)+s_i \rho T^{1/3}w_c}H(\tau_i,s_i)=&
\frac{1}{2\pi\I}\int_{\infty e^{-\pi\I/3}}^{\infty e^{\pi\I/3}}e^{t^3/3-s_it-\tau_it^2}dt +\Or(T^{-1/3}) \\
=&\Ai(\tau_i^2+s_i)e^{-2\tau_i^3/3-\tau_i s_i}+\Or(T^{-1/3}),\\
w_ce^{-TF_j(w_c)-s_i \rho T^{1/3}w_c}G(\tau_i,s_i)
=&\frac{1}{2\pi\I}\int_{\infty e^{-2\pi\I/3}}^{\infty e^{2\pi\I/3}}e^{-s^3/3+s_is+\tau_is^2}dt +\Or(T^{-1/3})\\\
=&\Ai(\tau_i^2+s_i)e^{2\tau_i^3/3+\tau_i s_i}+\Or(T^{-1/3}).
\end{aligned}
\end{equation}
In the above we used Appendix A in~\cite{BFP09} to derive Airy identities.
In the case where $s_i>0$, one also gets the following exponential decay: let $\kappa' >0$ be given. Then, as we can
modify the contours ${\cal C}$ and ${\cal C}'$ so that $\Re(w-w_c)<-\frac{\kappa'}{T^{1/3}\rho }$ while $\Re(z-w_c)>\frac{\kappa'}{T^{1/3}\rho }$ one gets that
\begin{equation}
\begin{aligned}
\Big|\frac{1}{w_c}e^{TF_i(w_c)+s_i\rho T^{1/3}w_c}H(\tau_i,s_i)\Big|&\leq C\frac{e^{-\kappa' s_i}}{T^{1/3}},\\
\Big|w_ce^{-TF_j(w_c)-s_i \rho T^{1/3}w_c}G(\tau_i,s_i)\Big|&\leq C\frac{e^{-\kappa' s_i}}{T^{1/3}}. \label{expdecHG}
\end{aligned}
\end{equation}

This ensures that
\begin{equation}\lim_{N \to \infty}\rho T^{1/3}\frac{Z(i)}{Z(j)}K_N^1(n_i,x_i; n_j,
x_j)=\int_{0}^{\infty}e^{-\lambda (\tau_i-\tau_j)}\Ai(s_i+\tau_i^2+\lambda)\Ai(s_j+\tau_j^2+\lambda)d\lambda
\end{equation}
in the trace-norm class (one can choose $\kappa' =\max \{|\tau_i|, i=1, \ldots, N\}+\kappa)$.
\end{proof}

We also need to consider the asymptotics of $\rho T^{1/3}\Psi_{n_i,n_j}(x_i,x_j)\frac{Z(i)}{Z(j)}.$

\begin{prop} \label{Prop: psi}
For $|s_i-s_j|$ in a bounded interval, it holds
\begin{equation}\rho T^{1/3}\Psi_{n_i,n_j}(x_i,x_j)\frac{Z(i)}{Z(j)}=\frac{1}{\sqrt{4\pi (\tau_j-\tau_i)} }\exp\left(-\frac{(s_j-s_i)^2}{4(\tau_j-\tau_i)}\right)+\Or(T^{-1/3}).
\end{equation}
Furthermore for any $\kappa>0$, there exists a $T_0$ large enough such that
\begin{equation}
\Big|\rho T^{1/3}\Psi_{n_i,n_j}(x_i,x_j)\frac{Z(i)}{Z(j)}\Big|\leq  C e^{-\kappa|s_i-s_j|+(\tau_i-\tau_j)(s_i+s_j)/2}. \label{majoPsi}
\end{equation}
for all $s_i,s_j\in\R$ and $T\geq T_0$. The constant $C$ is uniform in $T\geq T_0$ and $s_i,s_j$.
\end{prop}\begin{proof}[Proof of Proposition~\ref{Prop: psi}]
The asymptotics of $\Psi $ are again analyzed through a saddle point argument.
Consider
\begin{equation}
f(w):=\chi w -(\rho \gamma)^{-1}\ln (w+1-\eta).
\end{equation}
Then
\begin{equation}
\Psi_{n_i,n_j}(x_i,x_j)=\frac{1}{2\pi\I}\oint_{{\cal C}'} e^{2(\tau_j-\tau_i) f(w)T^{2/3}+\rho (s_j-s_i)w T^{1/3}}dw.
\end{equation}
The critical point is again $w_c=\eta-\frac{\gamma}{1+\gamma}$ and $f''(w_c)=\rho^2>0.$
We choose the contour to be the circle centered at $\eta -1$ and passing through $w_c$: this is a steep descent path. Making the change of variables $w=w_c+\frac{\I t}{\rho T^{1/3}}$ one gets that
\begin{equation}
\begin{aligned}
\rho T^{1/3}\Psi_{n_i,n_j}(x_i,x_j)e^{-f(w_c)}&=\frac{1}{2\pi }\int_{\R}e^{-t^2(\tau_j-\tau_i)+\I t(s_j-s_i)}dt+\Or(T^{-1/3})\\
&=\frac{1}{\sqrt{4\pi (\tau_j-\tau_i)} }\exp\left(-\frac{(s_j-s_i)^2}{4(\tau_j-\tau_i)}\right)+\Or(T^{-1/3}).
\end{aligned}
\end{equation}
Thus one obtains
\begin{equation}
\begin{aligned}
&\rho T^{1/3}\Psi_{n_i,n_j}(x_i,x_j)e^{T(F_i(w_c)-F_j(w_c))+(s_i-s_j)T^{1/3}w_c+2(\tau_i^3-\tau_j^3)/3+\tau_i s_i-\tau_js_j}\\
&=\frac{1}{\sqrt{4\pi (\tau_j-\tau_i)} }\exp\left(-\frac{(s_j-s_i)^2}{4(\tau_j-\tau_i)}+\frac23(\tau_i^3-\tau_j^3)+\tau_i s_i-\tau_js_j\right)+\Or(T^{-1/3}).
\end{aligned}
\end{equation}
The exponential decay for large $s_j-s_i$ is obtained as in~\cite{BFP09}, Lemma 21. For large positive $s_j-s_i$ we can modify the contour so that it lies to the left of $w_c$: $\Re (w-w_c)<-\kappa' /(\rho T^{1/3})$ for any arbitrary $\kappa' >0.$ For large negative $s_j-s_i$, the contour is modified in the following way: it is again a circle centered at $\eta-1$ but passing through $\widetilde w_c:=w_c+\kappa'/(\rho T^{1/3})$. It is a simple computation to check that $f$ decreases along this modified contour and ensures that
\begin{equation}
\Big|\rho T^{1/3}\Psi_{n_i,n_j}(x_i,x_j)e^{T(F_i(w_c)-F_j(w_c))+(s_i-s_j)T^{1/3}w_c}\Big|\leq C e^{-\kappa' |s_i-s_j|}.
\end{equation}
The complete details to derive (\ref{majoPsi}) from the above estimate is given in Lemma 21 in~\cite{BFP09}
(choosing again $\kappa' \geq \max |\tau_i|+\kappa$).
 \end{proof}
Combining Proposition~\ref{prop: asym+decayKN}, Proposition~\ref{Prop: psi} and the definitions of Section~\ref{subsectDefProc} yields part (a1) of Proposition~\ref{ThmOneSidedLPPbasic}.

\subsubsection{The case where $\eta=\frac{\gamma}{1+\gamma}$}
The rescaling is still given by (\ref{rescaling}). In this case, $w_c=0$. We recall that the contour ${\cal C}'$ has to encircle the pole $w=0$. Thus to get the exponential decay for large positive $s_i$, one needs to consider a different conjugation of the kernel. Indeed it is no longer possible to deform the contour ${\cal C}'$ so that it lies to the left of $w_c$. On the other hand, it is a well known fact that conjugation does not impact on the correlation functions of a determinantal random point process.

Let then $\delta>0$ be given. Define
\begin{equation}
Z(i, \delta):= \exp\left(\frac23 \tau_i^3+\tau_i s_i+TF_i(w_c)+s_i\rho T^{1/3}\left(w_c+\frac{\delta}{\rho T^{1/3}}\right)\right).
\end{equation}

\begin{prop}  \label{Prop:wc=0}
Uniformly for $s_i,s_j$ in a bounded interval, it holds
\begin{equation}
\begin{aligned}
&\lim_{N \to \infty} \rho T^{1/3}\frac{Z(i, \delta)}{Z(j, \delta)}K_N^1(n_i,x_i; n_j,x_j)\\
&= \Ai(s_i+\tau_i^2)e^{(s_i-s_j)\delta} \left (e^{ -2\tau_j^3/3-\tau_j s_j}-\int_0^{\infty}\Ai(\tau_j^2+s_j+x)e^{\tau_j x}dx\right)\\
&+e^{(s_i-s_j)\delta}\int_{0}^{\infty}e^{-\lambda (\tau_i-\tau_j)}\Ai(s_i+\tau_i^2+\lambda)\Ai(s_j+\tau_j^2+\lambda)d\lambda
+\Or(T^{-1/3}).
\end{aligned}
\end{equation}
Furthemore, for any $\kappa>0$, there exists a $T_0$ large enough such that
\begin{equation}
\Big|\rho T^{1/3}\frac{Z(i, \delta)}{Z(j, \delta )}K_N^1(n_i,x_i; n_j,x_j)\Big|\leq e^{-\kappa (s_i+s_j)}
\end{equation}
for all $s_i,s_j\in\R$ and $T\geq T_0$. The constant $C$ is uniform in $T\geq T_0$ and $s_i,s_j$.
\end{prop}
\begin{proof}[Proof of Proposition~\ref{Prop:wc=0}]
Define then
\begin{equation}
\begin{aligned}
\widetilde H(\tau_i,s_i)&=\frac{(\rho T^{1/3})^2}{2\pi\I}
 \oint_{\widetilde{\cal C}} e^{-Tz -2\tau_i \chi T^{2/3}z} z
\frac{( z+1-\eta  )^{\frac{T}{(1+\gamma)^2} +2\tau_i \frac{T^{2/3}}{\rho \gamma}}}{ ( z-\eta  )^{T (\frac{\gamma}{1+\gamma})^2}} e^{-s_i T^{1/3}\rho z}dz,\\
 \widetilde G(\tau_i,s_i)&=\frac{1}{2\pi\I}\oint_{\cal C'} e^{Tw +2\tau_i \chi T^{2/3}w} \frac1w
 \frac{ ( w-\eta )^{T (\frac{\gamma}{1+\gamma})^2}}{( w+1-\eta)^{\frac{T}{(1+\gamma)^2} +2\tau_i \frac{T^{2/3}}{\rho \gamma}} } e^{s_i T^{1/3}\rho w }dw.
\end{aligned}
\end{equation}
We can now perform the saddle point analysis of the above kernels.
We use the same contours as in the proof of Proposition~\ref{prop: asym+decayKN} up to the following modifications: the contours $\widetilde{\cal C}$ and ${\cal C}'$ are deformed in a $T^{-1/3}$ neighborhood of $w_c=0$ so that ${\cal C}'$ encircles $0$ lying to the left of $\widetilde w_c:=w_c+\delta/(\rho T^{1/3})$ and $\widetilde{\cal C}$ remains to the right of $\widetilde w_c.$
Furthermore we can assume that the distance of these contours to $\widetilde w_c$ is at least $\kappa'/T^{1/3}.$
From the preceding, one easily gets for bounded $s_i$ that
\begin{equation}
\begin{aligned}
Z(i,\delta) \widetilde H(\tau_i,s_i)&= e^{s_i\delta + \frac{2}{3} \tau_i^3+\tau_i s_i}\frac{1}{2\pi\I}
\int_{\infty e^{-\pi\I/3}}^{\infty e^{\pi\I/3}}te^{t^3/3-s_it-\tau_it^2}dt +\Or(T^{-1/3})\\
&=e^{s_i\delta }\left( -{\Ai}'(\tau_i^2+s_i)+\tau_i\Ai(\tau_i^2+s_i)\right). \label{asytildeH}
\end{aligned}
\end{equation}
Similarly
\begin{equation}
\frac{\widetilde G(\tau_i,s_i)}{Z(i,\delta)}= e^{-s_i\delta -2\tau_i^3/3-\tau_i s_i}\frac{1}{2\pi\I}
\int_{\infty e^{-2\pi\I/3}}^{\infty e^{2\pi\I/3}}\frac{1}{t}e^{-t^3/3+s_it+\tau_it^2}dt +\Or(T^{-1/3}),
\end{equation}
where the contour passes to the right of $0$. Using again Appendix A in~\cite{BFP09}
\begin{equation}
\frac{\widetilde G(\tau_i,s_i)}{Z(i,\delta)}=e^{-s_i\delta}\left (e^{ -2\tau_i^3/3-\tau_i s_i}-\int_0^{\infty}\Ai(\tau_i^2+s_i+x)e^{\tau_i x}dx\right)+\Or(T^{-1/3}).
\end{equation}
The exponential decay (as in (\ref{expdecHG})) for large positive $s_i$ follows from the fact that the $w-$contour (resp.\  $z$-contour) lies to the left (resp.\ right) of $\widetilde w_c$ and with a distance at least $\kappa'/T^{1/3}.$ Again one shall choose $\kappa'\geq \max_i |\tau_i|+\kappa.$

Finally, to derive the asymptotic correlation kernel in the case where $s_i,s_j$ lie in a fixed bounded set, we use the simple algebra:
\begin{equation}
\begin{aligned}
&\int_{0}^{\infty}e^{-\lambda (t-s)}d\lambda \frac{1}{(2\pi\I)^2}
\int_{\infty e^{-2\pi\I/3}}^{\infty e^{2\pi\I/3}}ds\int_{\infty e^{-\pi\I/3}}^{\infty e^{\pi\I/3}}dt
\left(\frac{t}{s}-1\right)\frac{e^{t^3/3-s_i t-\tau_i t^2}}{e^{s^3/3-s_j s-\tau_j s^2}} \\
&= \frac{1}{(2\pi\I)^2}
\int_{\infty e^{-2\pi\I/3}}^{\infty e^{2\pi\I/3}}ds\int_{\infty e^{-\pi\I/3}}^{\infty e^{\pi\I/3}}dt
\frac{1}{s}\frac{e^{t^3/3-s_i t-\tau_i t^2}}{e^{s^3/3-s_j s -\tau_j s^2}},
\end{aligned}
\end{equation}
yielding the asymptotic of $K_N^1$ given in Proposition~\ref{Prop:wc=0}.
\end{proof}

The asymptotic analysis of $\Psi$ for $w_c=0$ is almost unchanged from the last subsection.
For bounded $s_i-s_j$, one gets that
\begin{multline}
\rho T^{1/3}\Psi_{n_i,n_j}(x_i,x_j)\frac{Z(i, \delta)}{Z(j,\delta)}\\
=\frac{e^{(s_i-s_j)\delta}}{\sqrt{4\pi (\tau_j-\tau_i)}}\exp\left(-\frac{(s_j-s_i)^2}{4(\tau_j-\tau_i)}+\frac23(\tau_i^3-\tau_j^3)+\tau_i s_i-\tau_js_j\right)+\Or(T^{-1/3}).
\end{multline}
To get the exponential decay (as in (\ref{majoPsi})), one simply deforms the contour  to the right or left of $\widetilde w_c$ depending on the sign of $s_i-s_j$.

Combining the above with Proposition~\ref{Prop:wc=0} and definitions in Section~\ref{subsectDefProc} yields part (a2) of Proposition~\ref{ThmOneSidedLPPbasic}, by using the fact that
\begin{equation}
\det(\Id-\chi_s K_1 \chi_s)_{L^2(\{\tau_1,\ldots, \tau_m\}\times \R)}=\det(\Id-\chi_s K^b_1 \chi_s)_{L^2(\{\tau_1,\ldots, \tau_m\}\times \R)},\end{equation}
where $K_1^b(s,x;t, y)=e^{b(x-y)}K_1(s,x;t, y)$ for any $b$ in a compact interval.

\subsubsection{The case where $\eta<\frac{\gamma}{1+\gamma}$}
Let $m$ be a given integer. We now consider the asymptotic joint distribution
\begin{equation}
\Pb\left(\bigcap_{k=1}^m\left\{L_1(T,\gamma_k^2 T)\leq \left(\frac{\gamma_k^2}{\eta}+\frac{1}{1-\eta}\right)T +s_k T^{1/2}\right\}\right).
\end{equation}
For $k=1, \ldots, m$ we set
\begin{equation}
x_k=\left(\frac{\gamma_k^2}{\eta}+\frac{1}{1-\eta}\right)T +s_k T^{1/2},\quad n_k=\gamma_k^2T,\quad
c_k=\sqrt{\frac{\gamma_k^2}{\eta^2}-\frac{1}{(1-\eta)^2}}.
\end{equation}
We recall that the correlation kernel $K_N$ is defined in (\ref{K_N1}) and let us set $Z(i):=\dfrac{(-\eta)^{\gamma_i^2T}}{(1-\eta)^T}$, for a given small $\delta>0$.
\begin{prop}
\label{Prop: gauss}
For $s_i- s_j$ in a bounded interval, it holds
\begin{equation}
\begin{aligned}
&\lim_{N \to \infty}\sqrt{T} \frac{Z(i)}{Z(j)}e^{(s_i-s_j)\delta} K_N(n_i,x_i;n_j,x_j)\\
&=\frac{e^{(s_i-s_j)\delta}}{\sqrt{2\pi c_i^2}}\exp\left(-\frac{s_i^2}{ 2c_i^2}\right)
-\Id_{[\gamma_i<\gamma_j]}\frac{e^{(s_i-s_j)\delta}}{\sqrt{2\pi(c_j^2-c_i^2)}}\exp\left(-\frac{(s_i-s_j)^2}{2(c_j^2-c_i^2)}\right).
\end{aligned}
\end{equation}
Furthermore, for any $\kappa>0$, there exists a $T_0$ large enough such that
\begin{equation}
\left|\sqrt T \frac{Z(i)}{Z(j)}e^{(s_i-s_j)\delta}K_N(n_i,x_i;n_j,x_j)\right|\leq C e^{-\kappa |s_i-s_j|},
\end{equation}
for all $s_i,s_j\in\R$ and $T\geq T_0$. The constant $C$ is uniform in $T\geq T_0$ and $s_i,s_j$.
\end{prop}

\begin{proof}[Proof of Proposition~\ref{Prop: gauss}]
Define
\begin{equation}
f_j(w):=\left(\frac{\gamma_j^2}{\eta}+\frac{1}{1-\eta}\right) w+\gamma_j^2 \ln (w-\eta)-\ln (w+1-\eta).
\end{equation}
Then
\begin{equation}
K_N^1(n_i,x_i; n_j,x_j)= \frac{1}{(2\pi\I)^2}\oint_{{\cal C}}dz\oint_{{\cal C}'}dw\frac{z}{w}\frac{1}{w-z}e^{T(f_j(w)-f_i(z))+T^{1/2}(s_jw-s_iz)}.
\end{equation}
It is easy to check that the exponential term $f_j$ admits two critical points $w_c=0$ and $w_c^-=\frac{\eta^2-\gamma_j^2(1-\eta)^2}{\gamma_j^2(1-\eta)+\eta}<0$.
The critical point $w_c=0$ satisfies $f''(0)=\frac{1}{(1-\eta)^2}-\frac{\gamma_j^2}{\eta^2}<0.$
The steep descent (resp.\  ascent) path for $f_j$ should pass through $w_c^-$ (resp.\  $w_c=0$). Nevertheless the contour ${\cal C}'$ has to encircle the critical point $w_c$. To deal with this difficulty, we separate the contribution of the pose at $w=0$ and the pole at $w=\eta-1$. This will modify a little bit the saddle point analysis, which turns out to be similar to the analysis in Section~3 of~\cite{BBP06}.

Computing the residue at $w=0$, one gets that
\begin{equation}
\begin{aligned}
\label{corrkernelSchur3}
K_N(n_i,x_i; n_j,x_j)&=
\frac{1}{(2\pi\I)^2}\oint_{{\cal C}}dz\oint_{\mathcal{C''}}dw \frac{e^{wx_j-zx_i}}{w-z}\frac{(z+1-\eta)^{T}}{(w+1-\eta)^{T}}\frac{z}{w}\frac{(w-\eta)^{\gamma_j^2T}}{(z-\eta)^{\gamma_i^2T}}\\
&-\Psi_{n_i,n_j}(x_i,x_j) -\frac{1}{(2\pi\I)}\oint_{{\cal C}}dze^{-zx_i}\frac{(z+1-\eta)^{T}}{(1-\eta)^{T}}\frac{(-\eta)^{\gamma_j^2T}}{(z-\eta)^{\gamma_i^2T}},
\end{aligned}
\end{equation}
where the contour ${\cal C}''$ does not encircle the pole $w=0.$

Let us first consider
\begin{equation}
K_N^2(n_i,x_i;n_j,x_j):=\frac{1}{(2\pi\I)}\oint_{{\cal C}} dze^{-zx_i}\frac{(z+1-\eta)^{T}}{(1-\eta)^{T}}\frac{(-\eta)^{\gamma_j^2T}}{(z-\eta)^{\gamma_i^2T}}.
\end{equation}
For ease we assume that $\eta<1/2$ so that $\gamma_j>1$, $j=1, \ldots,m$.
Set the contour ${\cal C}_1=\{z=\I t, |t|\leq 2\}$. Then $2\Re(f_i(\I t))= \gamma_i^2\ln (t^2+\eta^2)-\ln(t^2+(1-\eta)^2)$ so that
\begin{equation}
\frac{d}{dt}\Re (f_i(\I t))=t\left( \frac{\gamma_i^2}{t^2+\eta^2}-\frac{1}{t^2+(1-\eta)^2}\right)=t\frac{(\gamma_i^2-1)t^2+\gamma_i^2(1-\eta)^2-\eta^2}{(t^2+\eta^2)(t^2+(1-\eta)^2)}
\end{equation}
has the same sign as $t$. Thus ${\cal C}_1$ is a steep descent path for $-f_i$ with maximum at $z=0$.
We complete ${\cal C}_1$ by the contour ${\cal C}_2=\{z=\eta +\sqrt{4+\eta^2}e^{\I \theta}, 0\leq \theta\leq \theta_0\}$
where $\theta_0$ is defined by $\eta +\sqrt{4+\eta^2}e^{\I \theta_0}=2\I t.$
Then setting \mbox{$z=\eta +\sqrt{4+\eta^2}e^{\I \theta}$},
\begin{equation}
\Re \frac{d}{d\theta}f_i(z)=-\Im(z)\left(\frac{\gamma_i^2}{\eta}+\frac{1}{1-\eta}-\frac{1}{|z+1-\eta|^2}\right),
\end{equation}
and there exists a $c>0$ such that $\frac{\gamma_i^2}{\eta}+\frac{1}{1-\eta}-\frac{1}{|z+1-\eta|^2}\geq c$ along ${\cal C}_2$, so that also ${\cal C}_2$ is a steep descent path for $-f_i$.

Then, for bounded $s_i$ we obtain
\begin{equation}
\lim_{N \to \infty}\sqrt{T}\frac{Z(i)}{Z(j)}K_N^2(n_i,x_i;n_j,x_j)=\frac{1}{2\pi c_i}
\int_{\R}e^{-t^2/2+\I ts_i/c_i}dt=\frac{1}{\sqrt{2\pi}c_i}e^{-\frac{s_i^2}{2c_i^2}}.
\end{equation}

Next, we consider
\begin{equation}
K_N^3(n_i,x_i;n_j,x_j):=\frac{1}{(2\pi\I)^2}\oint_{{\cal C}}dz\oint_{\mathcal{C''}}dw \frac{e^{wx_j-zx_i}}{w-z}\frac{(z+1-\eta)^{T}}{(w+1-\eta)^{T}}\frac{z}{w}\frac{(w-\eta)^{\gamma_j^2T}}{(z-\eta)^{\gamma_i^2T}}
\end{equation}
for some constant $C>0$.
We set ${\cal C}''=\{|w_c^-|e^{\I\theta}, \theta \in [0, 2\pi]\}$ so that $w_c^--\eta=-\gamma_j (1-\eta)$ ensuring that
\begin{equation}
\frac{d}{d\theta} \Re f_j(|w_c^-|e^{\I\theta})<-|w_c^-|\sin \theta\left(\frac{\gamma_j^2}{\eta}- \frac{\eta}{(1-\eta)^2}\right)<-C\sin \theta,
\end{equation}
for $0\leq \theta \leq \pi$. Thus $\Re f_j$ achieves its maximum on ${\cal C}''$ at $w=w_c^-$. Now a simple computation shows
that $f_j(0)-f_j(w_c^-)=\int_{w_c^-}^0 f_j'(x)dx>0.$
Thus, one has that
\begin{equation}
\left|\sqrt T \frac{Z(i)}{Z(j)}K_N^3(n_i,x_i;n_j,x_j)\right|\leq C e^{-cT},
\end{equation}
for some constants $C,c>0$.  Thus, for bounded $s_i$ we obtain
\begin{equation}
\lim_{N \to \infty}\sqrt{T}\frac{Z(i)}{Z(j)}K_N^3(n_i,x_i;n_j,x_j)=0.
\end{equation}

To consider large positive $s_i$, we consider the conjugated kernel
\begin{equation}
K_N^1(n_i,x_i;n_j,x_j)e^{(s_i-s_j)\delta},
\end{equation}
for some $\delta>0$ small.
The contours ${\cal C}_1$ is modified so that it passes to the right of $(\delta+\kappa)/T^{1/2}$. The $w-$contour passes to the left of $(\delta-\kappa)/T^{1/2}$. The exponential decay for large positive $s_i$ follows.

The analysis of $\Psi_{n_i,n_j}$ is similar to those of the preceding sections. The exponential term to be considered is
\begin{equation}
g(w):=(\gamma_i^2-\gamma_j^2)\left (\frac{w}{\eta}+\ln (w-\eta)\right).
\end{equation}
$g$ has a single critical point $w_c=0$ with $g''(0)=-\frac{\gamma_i^2-\gamma_j^2}{\eta^2}=c_j^2-c_i^2$.
Consider the contour $w=\eta(1+e^{\I\theta})$. Then the leading (i.e., non exponentially negligible) contribution to $\Psi_{n_i,n_j}$ comes from a neighborhood of width $T^{-1/2}$ of $w=0$.
Thus for bounded $|s_j-s_i|$ we deduce that
\begin{equation}
\begin{aligned}
\lim_{T\to\infty}\sqrt{T}\frac{Z(i)}{Z(j)}\Psi_{n_i,n_j}(x_j,x_i)&=
\frac{1}{2\pi}\int_{\R}dt e^{-(c_j^2-c_i^2)t^2/2+\I t(s_i-s_j)}\\ &=
\frac{1}{\sqrt{2\pi(c_j^2-c_i^2)}}\exp\left(-\frac{(s_i-s_j)^2}{2(c_j^2-c_i^2)}\right).
\end{aligned}
\end{equation}
For large $|s_j-s_i|$, we consider the conjugated kernel
\begin{equation}
K_N(n_i,x_i;n_j,x_j)e^{(s_i-s_j)\delta}.
\end{equation}
Depending on the sign of $s_i-s_j$ we modify the contour to be the circle of ray $\delta \pm 2\kappa'/\sqrt T$. This ensures the exponential decay for large $|s_i-s_j|$.
\end{proof}

The above Proposition~\ref{Prop: gauss} has the required asymptotic results needed to conclude part (b) of Proposition~\ref{ThmOneSidedLPPbasic}.

\subsection{Coupling Lemmas}\label{coupling_lemmas}
The following technical lemmas provide a basis for the coupling arguments necessary in our proof of Theorem~\ref{ThmTwoSidedLPP}. They provide generalizations of Lemma 4.1 and Lemma 4.2 of~\cite{BC09} from one-point functions to n-point functions. The proofs, however, are almost identical.

For the purpose of the lemmas let $X_n$ and $\tilde X_n$ take values in $\R^{k}$. We say that $X_n\geq \tilde X_n$ if, with probability one, every coordinate of $X_n$ is greater than or equal to the corresponding coordinate of $\tilde X_n$. We say that $X_n\Rightarrow F$, where $F$ is a distribution function on $\R^k$ if for all $\e>0$ and $(s_1,\ldots, s_{k})$ a continuity point of $F$, there exists an $N=N(\e,s_1,\ldots,s_{k})$ such that for all $n>N$,
\begin{equation}
\left|\Pb\left(\cap_{i=1}^k\{X_n^i\leq s_i\}\right)- F(s_1,\ldots, s_{k})\right|<\e.
\end{equation}
Finally we say that $X_n-\tilde X_n$ converges in probability to zero if for all $\e>0$, $\lim_{n\to\infty}\Pb(\|X_n-\tilde X_n\|_{\infty}>\e)\rightarrow 0$ (the infinity norm is just that max over all finitely many coordinates: $\|X\|_{\infty}=\max_{1\leq i\leq k} |X^i|$).

In the three lemmas below we assume all random variables are $\R^k$ valued.

\begin{lem}\label{stoch_dom_convergence}
If $X_n\geq\tilde X_n$ and $X_n\Rightarrow D$ as well as $\tilde X_n\Rightarrow D$, then $X_n-\tilde{X}_n$ converges to zero in probability. Conversely, if  $X_n\geq\tilde X_n$, $\tilde X_n\Rightarrow D$ and $X_n-\tilde{X}_n$ converges to zero in probability then $X_n\Rightarrow D$ as well.
\end{lem}

For vectors $X$ and $Y$ in $\R^k$ we define $Z=\max(X,Y)$ to be the coordinate-wise maximum (i.e., $Z^i=\max(X^i,Y^i)$ for $i=1,\ldots, k$).

\begin{lem}\label{max_lemma}
Assume $X_n\geq \tilde X_n$ and  $X_n\Rightarrow D_1$  as well as
$\tilde X_n\Rightarrow D_1$; and similarly
 $Y_n\geq \tilde Y_n$ and  $Y_n\Rightarrow D_2$  as well as $\tilde Y_n\Rightarrow D_2$. Let $Z_n=\max(X_n,Y_n)$ and $\tilde Z_n=\max(\tilde X_n,\tilde Y_n)$. Then if $\tilde Z_n\Rightarrow D_3$, we also have $Z_n\Rightarrow D_3$.
\end{lem}

\begin{lem}\label{multipoint_marginal_lemma}
Assume $X_n\geq \tilde X_n$ and  $X_n\Rightarrow D_1$  as well as
$\tilde X_n\Rightarrow D_1$; and similarly $Y_n\geq \tilde Y_n$ and  $Y_n\Rightarrow D_2$  as well as $\tilde Y_n\Rightarrow D_2$. Then if \mbox{$(\tilde X_n,\tilde Y_n)\Rightarrow F$} (a $2k$-dimensional distribution function) so does $(X_n, Y_n)\Rightarrow F$. More generally this also applies to $m$ sequences of random variables under the same hypotheses on each sequence and on their $m$-point joint distribution function limit.
\end{lem}

\begin{proof}[Proof of Lemma~\ref{stoch_dom_convergence}]
This proof is a straight forward generalization of the proof of Lemma 4.1 of~\cite{BC09} and hence we will not reproduce it. All inequalities in the original proof should now be considered as holding true coordinate-wise, all absolute values should be replaced by $\ell_\infty$ norms (on $\R^k$). The $\e$ sized blocks used should be replaced by $k$ dimensional $\e$ boxes, and all intervals should interpreted as boxes in $\R^k$. Other than these changes, the proof goes through word for word.
\end{proof}

\begin{proof}[Proof of Lemma~\ref{max_lemma}]
 Again this proof is word for word the same as Lemma 4.2 of~\cite{BC09}, with the modified interpretations of notation noted above.
\end{proof}

\begin{proof}[Proof of Lemma~\ref{multipoint_marginal_lemma}]
Lemma~\ref{stoch_dom_convergence} shows that $X_n-\tilde X_n$ and likewise $Y_n-\tilde Y_n$  converges in probability to zero. This implies that for all $\e>0$, using the triangle inequality and the union bound,
\begin{equation}
\Pb(\|(X_n,Y_n)-(\tilde X_n,\tilde Y_n)\|_\infty>\e)\leq \Pb(\|X_n-\tilde X_n\|_\infty>\e)+\Pb(\|Y_n-\tilde Y_n\|_\infty>\e),
\end{equation}
which, by Lemma~\ref{stoch_dom_convergence} goes to zero as $n\to\infty$. This immediately implies that the joint $(X_n,Y_n)$ converge to the same distribution as $(\tilde X_n, \tilde Y_n)$. Lemma~\ref{max_lemma} is, in fact a corollary of this result. The generalization follows by the exact same argument as above.
\end{proof}

\subsection{Proof of Theorem~\ref{ThmTwoSidedLPP}}\label{proof_two_sided_section}
In Section~\ref{subsectProofOneSided} we proved Proposition~\ref{ThmOneSidedLPPbasic} directly from asymptotic analysis of the Schur Process. From that theorem we will, using the three lemmas above and the slow decorrelation result of Proposition~\ref{ThmSlowDec}, provide proofs of Proposition~\ref{ThmOneSidedLPP} and Theorem~\ref{ThmTwoSidedLPP}.

\begin{proof}[Proof of Proposition~\ref{ThmOneSidedLPP}]
The proof is based on Proposition~\ref{ThmOneSidedLPPbasic} together with slow-decorrelation phenomenon (Proposition~\ref{ThmSlowDec}, see also Remark~\ref{decorrelation_remark}).

Consider first cases (a1) and (a2). Denote by $\tilde\tau_k$ the number such that $(x(\tilde \tau_k),y(\tilde\tau_k))$ (defined in (\ref{scaling0})) and $(x(\tau_k,\theta_k),y(\tau_k,\theta_k))$ (defined in (\ref{scaling1}) belongs to the same characteristic line. Moreover, notice that the projection along the characteristic direction for $\tau=0$ to the line $y=\frac{\gamma^2}{(1+\gamma)^2} T$ is obtained by choosing
\begin{equation}
\theta_k = 2\tau_k \frac{(1+\gamma)^{4/3}}{\gamma^{2/3}(1+\gamma^2)} T^{2/3-\nu}
\end{equation}
in (\ref{scaling1}). However, the slope of the characteristic line passing by $(x(\tau_k,\theta_k),y(\tau_k,\theta_k))$ differs from the slope of the characteristic line for $\tau_k=0$ by just $\Or(T^{-1/3})$. Therefore, as $\nu<1$, $\tilde\tau_k=\tau_k+\Or(T^{\nu-1})\to \tau_k$ as $T\to\infty$. Also, due to slow decorrelations, the fluctuation of $L_1(x(\tau_k,\theta_k),y(\tau_k,\theta_k))$ differs from the fluctuation of $L_1(x(\tilde \tau_k),y(\tilde \tau_k))$ by $o(T^{1/3})$, whose differs from the fluctuations of $L_1(x(\tau_k),y(\tau_k))$ again by $o(T^{1/3})$. Thus Proposition~\ref{ThmOneSidedLPP}~(a1) and (a2) follows.

The case (b) is even simpler. In that case, the point $(\theta_k T,\gamma^2_k\theta_k T)$ is on the same characteristic line as $(T,\gamma^2_k T)$, at a distance $\Or(T)$. Therefore by Proposition~\ref{ThmSlowDec}~(b) the fluctuations of $L_1(\theta_k T,\gamma^2_k\theta_k T)$ and $L_1(T,\gamma^2_kT)$ differs only by $o(T^{1/2})$, from which Proposition~\ref{ThmOneSidedLPP}~(b) follows.
\end{proof}

\begin{proof}[Proof of Theorem~\ref{ThmTwoSidedLPP}]
We follow the method of~\cite{BC09} and define two coupled random vectors $X$ and $Y$. $X$ is the vector of last passage times from $(0,0)$ to $(x(\tau_i,\theta_i),y(\tau_i,\theta_i))$, $1\leq i \leq m$, with last passage paths forced to take a first step to the right, and $Y$ is the vector of last passage times with paths forced to take a first step up. Therefore, their coordinate-wise maximum \mbox{$Z=\max(X,Y)$} is the last passage times without any restrictions on the first step (i.e., $Z^i= L_2(x(\tau_i,\theta_i),y(\tau_i,\theta_i))$. A key observation is that $X$ and $Y$ are both marginally distributed as the last passage times for last passage percolation models with only one-sided boundary conditions (as opposed to the two-sided conditions we must consider for $Z$). In the case of $Y$ the one-sided boundary waiting time is exponential of mean $1/\eta$ and in the case of $X$ the boundary waiting time is exponential of mean $1/\pi$. But the coordinates must be flipped so as to conform to our definition of last passage percolation with one-sided boundary conditions (the boundary condition should appear on the left boundary, not the bottom). Depending on the regime of fluctuations we will be able to compare the random vectors to related but simplified vectors $\tilde X$ and $\tilde Y$ which have the same asymptotic limiting distribution but are strictly less than $X$ and $Y$. Then, using the coupling lemmas we will be able to show that $Z$ and $\tilde Z = \max(\tilde X,\tilde Y)$ have the same distribution limits as $T$ goes to infinity. However, as $\tilde X$ and $\tilde Y$ are simpler than $X$ and $Y$, we will be able to identify $\tilde Z$ exactly and hence determine the asymptotic multipoint distribution of $Z$, completing our proof.

As slightly different coupling arguments are necessary for each part of the theorem we will split the proof up according to the four cases of the theorem.

\subsubsection{Proof of Theorem~\ref{ThmTwoSidedLPP} part (a1)}
This case corresponds to both $X$ and $Y$ being last passage time vectors from last passage percolation models with one-sided boundary conditions of small enough mean so as to behave asymptotically the same as the corresponding models without boundary conditions. With this in mind we define $\tilde{X}$ and $\tilde{Y}$, random vectors which are coupled to $X$ and $Y$ in terms of the underlying random last passage waiting times. Let $\tilde X = X$ and let $\tilde Y$ be the vector of last passage times defined by $\tilde Y^i=\tilde L_1( x(\tau_i,\theta_i),y(\tau_i,\theta_i))$. The new last passage time $\tilde L_1$ is the last passage time in a coupled model where the boundary waiting times (which are exponential with mean $1/\eta$) are multiplied by $\eta$ (hence making them distributed as exponentials of mean $1$). The key is that this random last passage time is coupled to the last passage time $L_1)$ since that they are based off of the same random waiting times. Additional, because of $\eta\leq 1$ it holds $\tilde L_1\leq L_1$. Therefore $\tilde Y \leq Y$ where the inequality is in terms of each coordinate separately. More trivially we also have that $\tilde X\leq X$.

In order to apply our coupling lemmas we must center and rescale $X,Y,\tilde X$ and $\tilde Y$ so that our new $X=(X^1,\ldots,X^m)$ equals the vector with coordinates
\begin{equation}
\frac{X^i - \ell(\tau_i,\theta_i,0)}{(1+\gamma)^{2/3}\gamma^{-1/3}T^{1/3}},\quad 1\leq i \leq m,
\end{equation}
with $\ell(\tau,\theta,0)$ given in (\ref{scaling1}), and likewise for the other variables. Under this centering and rescaling $\tilde X\leq X$ and Proposition~\ref{ThmOneSidedLPP} shows that both $\tilde X$ and $X$ converge in joint-distribution to the same $\mathcal{A}_2$ process as $T\to \infty$. Likewise $\tilde Y\leq Y$ by construction and $\tilde Y$ and $Y$ converge in joint-distribution to the same $\mathcal{A}_2$ process as well. Moreover, $\tilde Z = \max (\tilde X,\tilde Y)$ is (except for a single waiting time of zero at the origin, which is asymptotically irrelevant) the last passage time vector for a one-sided last passage percolation model with boundary waiting times with mean $1/\eta$. Proposition~\ref{ThmOneSidedLPP} shows that $\tilde Z$ converges in joint-distribution to the $\mathcal{A}_2$ process. Therefore, using Lemma~\ref{max_lemma} it follows that $Z=\max(X,Y)$ also converges to the $\mathcal{A}_2$ process, which is
 exactly what we needed to prove.

\subsubsection{Proof of Theorem~\ref{ThmTwoSidedLPP} part (a2)}
We are in the case of $\eta=\gamma(1+\gamma)^{-1}$ and $\pi>(1+\gamma)^{-1}$. As such, we can apply the exact same argument as in the proof of part (a1) above. The only difference is that the $\tilde Z$ process will now, as determined by Proposition~\ref{ThmOneSidedLPP}, converge to the $\mathcal{A}_{{\rm BM}\to 2}$ process. Therefore $Z$ will also converge in finite-distribution to the $\mathcal{A}_{{\rm BM}\to 2}$ process, which is, again, what we desired to show.

\subsubsection{Proof of Theorem~\ref{ThmTwoSidedLPP} part (a3) (see~\cite{BFP09})}
This proof is the subject of the recent paper~\cite{BFP09}. The coupling techniques employed for all of the other proof do not apply here. The heuristic explanation is that the last passage time comes from the competition of two sets of paths each of which goes along the boundary for distance of order $T^{2/3}$ and then enters the bulk. Because the range of the transversal fluctuations of a last passage path are of that order $T^{2/3}$, these sets of paths have non-trivial correlation, which is evident in that fact that they yield a different process, the $\mathcal{A}_{\rm stat}$ process.

\subsubsection{Proof of Theorem~\ref{ThmTwoSidedLPP} part (b)}
As before we write the last passage random variable $L_2(\theta_i T,\gamma_i^2\theta_i T)$ as $\max(X^i,Y^i)$ where $X^i$ and $Y^i$ are coupled last passage times, restricted to paths which step first right or up, respectively. We now couple $X^i$ with $\tilde X^i$ which is the last passage time when forced to stay along the bottom edge for a specific deterministic fraction of the path, and then depart into the bulk. Specifically we define $\tilde X^i$ to be the max of passage times over all paths which go distance
\begin{equation}\label{X_distance}
 \left(1-\frac{\gamma_i^2}{(\pi^{-1}-1)^2}\right)\theta_i T
\end{equation} and then take a step up. Likewise we define $\tilde Y^i$ to be the max of the passage times over all paths which go distance
\begin{equation}\label{Y_distance}
\left(\gamma_i^2-\frac{1}{(\eta^{-1}-1)^2}\right)\theta_i T
\end{equation}
and then take a step right. It is clear that $\tilde X^i\leq X^i$ and that $\tilde Y^i\leq Y^i$. What is not obvious is the choice of distances. In short, this is given by the solution to an optimization problem at the level of the law of large numbers (see~\cite{BBP06} for an explanation of this heuristic).

Define the following events: for $i\in\{1,\ldots,m_l+m_s\}$, set
\begin{equation}
E_i=\left\{L_2(\theta_i T,\gamma_i^2\theta_i T)\leq \left(\frac{\gamma_i^2}{\eta}+\frac{1}{1-\eta}\right)\theta_i T+s_i T^{1/2}\right\},
\end{equation}
while for $i\in\{1+m_l+m_s,\ldots,m\}$, set
\begin{equation}
E_i=\left\{L_2(\theta_i T,\gamma_i^2\theta_i T)\leq \left(\frac{1}{\pi}+\frac{\gamma_i^2}{1-\pi}\right)\theta_i T+s_i T^{1/2}\right\}.
\end{equation}
Notice that for $i\in\{m_l+1,\ldots,m_l+m_s\}$ both definitions are identical.
Likewise, define $\widetilde E_i$ except in place of $L_2(\theta_i T,\gamma_i^2\theta_i T)$ use $\tilde Z^i=\max(\tilde X^i, \tilde Y^i)$. Let us denote $L_2(\theta_i T,\gamma_i^2\theta_i T)=\max(X^i,Y^i)$ as $Z^i$. It is clear that $\tilde Z^i\leq Z^i$.
We claim that
\begin{equation}\label{claim1}
\lim_{T\to \infty} \Pb\bigg(\bigcap_{k=1}^{m} E_k\bigg) =\lim_{T\to \infty} \Pb\bigg(\bigcap_{k=1}^{m} \widetilde E_k\bigg).
\end{equation}

For $i\in \{1,\ldots ,m_l\}$ center and scale $X^i,\tilde X^i,Y^i,\tilde Y^i, Z^i$ and $\tilde Z^i$ by applying
\begin{equation}\label{eq3.67}
x\mapsto \frac{x-\left(\frac{\gamma_i^2}{\eta} + \frac{1}{1-\eta}\right)\theta_i T}{T^{1/2}}.
\end{equation}
It follows from the one-point fluctuation result of~\cite{BC09} that the centered and scaled $\tilde Z^i$ and $Z^i$ both converge in distribution to the same Gaussian random variable with variance
\begin{equation}\label{var1}
\theta_i\left[\frac{\gamma_i^2}{\eta^2}-\frac{1}{(1-\eta)^2}\right].
\end{equation}

For $i\in \{m_l+m_s+1,\ldots, m\}$ we center and scale with
\begin{equation}\label{eq3.67b}
x\mapsto \frac{x-\left(\frac{1}{\pi}+\frac{\gamma_i^2}{1-\pi} \right)\theta_i T}{T^{1/2}}.
\end{equation}
Then the centered and scaled $\tilde Z^i$ and $Z^i$ both converge in distribution to the same Gaussian random variable with variance
\begin{equation}\label{var2}
\theta_i\left[\frac{1}{\pi^2}-\frac{\gamma_i^2}{(1-\pi)^2}\right].
\end{equation}

For $i\in \{m_l+1,\ldots, m_l+m_s\}$ we center and scale as (\ref{eq3.67}).
Then the centered and scaled $\tilde Z^i$ and $Z^i$ both converge in distribution to the maximum of two Gaussian random variables with variances given by equations (\ref{var1}) and (\ref{var2}).

Since for every $i\in 1,\ldots, m$, the centered and scaled $Z^i$ and $\tilde Z^i$ converge to the same distributions, and since $\tilde Z^i\leq Z^i$, Lemma~\ref{multipoint_marginal_lemma} implies that the asymptotic joint distribution of the $Z^i$, and of the $\tilde Z^i$ converge to the same distribution. This proves the claim given in equation (\ref{claim1}).

Therefore it remains to show that the joint distribution of the $\tilde Z^i$ behaves as desired. The fluctuations given by $\tilde Z^i$ are a combination of the fluctuations from the boundary waiting times and from the bulk waiting times. However, since we scaled by $T^{1/2}$ and since the boundary and bulk fluctuations are independent, the bulk fluctuations have a prefactor of $T^{-1/6}$ and hence (by, for instance applying the Converging Together Lemma on page 89 of~\cite{Dur05}) only the boundary fluctuations contribute asymptotically. The covariance of these fluctuations depends on portion of the boundary which the $\tilde X^i$ and $\tilde Y^i$ depend upon. We can encode this covariance structure in terms of two independent Brownian motions (one for the left boundary and one for the bottom boundary). For $i\in \{1,\ldots ,m_l\}$ all of the fluctuations come from the left boundary. For $i\in \{m_l+m_s+1,\ldots, m\}$ all of the fluctuations come from the bottom boundary. For $i\in \{m_l+1,\ldots, m_l+m_s\}$ fluctuations come from the maximum of the left and bottom boundary Brownian motions. Writing down this joint distribution leads exactly to (\ref{bmeqn}).
\end{proof}

\begin{remark}
It is worth noting that in the proof of part~(b) above we did not, in fact, appeal to the analogous one-sided last passage percolation result of Proposition~\ref{ThmOneSidedLPP}. This is because we needed to establish the product structure and hence reduce everything to just processes along the boundary. As such the Brownian motion results of Proposition~\ref{ThmOneSidedLPPbasic} and~\ref{ThmOneSidedLPP} may, in fact, be proved directly in this manner (as they are corollaries of this result) and do not require the asymptotic analysis of the Schur process.
\end{remark}

\subsection{Proof of the TASEP height function theorem}\label{subsectProofTASEP}
\begin{proof}[Proof of Theorem~\ref{ThmTASEP}]
The connection between two-sided directed percolation and TASEP has been discussed in Section~\ref{sectConnection}. Consider first the cases (a1)-(a3). From (\ref{eqTASEPvsDP}) we have
\begin{multline}\label{eq3.71}
\Pb\bigg(\bigcap_{k=1}^m\{h_{T+\theta_k T^\nu}(X(\tau_k,\theta_k))\geq H(\tau_k,\theta_k,s_k)\}\bigg)\\
=\Pb\bigg(\bigcap_{k=1}^m\{L_2(x_k,y_k)\leq T+\theta_k T^\nu\}\bigg)
\end{multline}
with (difference of order 1, due to the integer parts which are not explicitly written since they are irrelevant in the asymptotics)
\begin{equation}
\begin{aligned}
x_k=\tfrac12(X(\tau_k,\theta_k)+H(\tau_k,\theta_k,s_k)),\\
y_k=\tfrac12(H(\tau_k,\theta_k,s_k)- X(\tau_k,\theta_k)),
\end{aligned}
\end{equation}
where $X$ and $H$ are defined in~(\ref{eq2.8}). Explicitly, by setting $\gamma:=(1-\xi)/(1+\xi)$, i.e., $\xi=(1-\gamma)/(1+\gamma)$, we have
\begin{equation}
\begin{aligned}
x_k&=\frac{1}{(1+\gamma)^2}(T+\theta_k T^\nu)+\tau_k\frac{2\gamma^{1/3}}{(1+\gamma)^{5/3}}T^{2/3}+(\tau_k^2-s_k)\frac{\gamma^{2/3}}{(1+\gamma)^{4/3}} T^{1/3},\\
y_k&=\frac{\gamma^2}{(1+\gamma)^2}(T+\theta_k T^\nu)-\tau\frac{2\gamma^{4/3}}{(1+\gamma)^{5/3}}T^{2/3}+(\tau_k^2-s_k)\frac{\gamma^{2/3}}{(1+\gamma)^{4/3}}T^{1/3}.
\end{aligned}
\end{equation}

Once the problem is rewritten in terms of directed percolation, the theorem is proven using Theorem~\ref{ThmTwoSidedLPP} and the slow decorrelation (see Proposition~\ref{ThmSlowDec}), i.e., we use a similar strategy of the proof of Proposition~\ref{ThmOneSidedLPP} starting from Proposition~\ref{ThmOneSidedLPPbasic}.

For the above given $(x_k,y_k)$, the limit shape (\ref{eqLimShapeRM}) gives us
\begin{equation}
x_k(1+\sqrt{y_k/x_k})^2=T+\theta_k T^\nu-s_k \frac{(1+\gamma)^{2/3}}{\gamma^{1/3}}T^{1/3}+\Or(1).
\end{equation}
Therefore,
\begin{multline}
\Pb\bigg(\bigcap_{k=1}^m\{L_2(x_k,y_k)\leq T+\theta_k T^\nu\}\bigg)\\
=\Pb\bigg(\bigcap_{k=1}^m\big\{L_2(x_k,y_k)\leq x_k(1+\sqrt{y_k/x_k})^2+s_k \frac{(1+\gamma)^{2/3}}{\gamma^{1/3}}T^{1/3}+\Or(1)\big\}\bigg).
\end{multline}
The fluctuations (with respect to the limit shape behavior) are, by the slow decorrelation theorem, the same as the fluctuations of the projection along the characteristic line on the line $x+y=\frac{1+\gamma^2}{(1+\gamma)^2}(T+\theta_k T^\nu)$. Exactly as in the proof of Proposition~\ref{ThmOneSidedLPP}, we can use an approximate characteristic line, namely the characteristic line for $\tau_k=s_k=0$. We look for $\tilde\tau_k$ such that
\begin{equation}
\begin{aligned}
x_k&=x(\tilde \tau_k,\theta_k)+r(1+\gamma)^{-2},\\
y_k&=y(\tilde \tau_k,\theta_k)+r\gamma^2(1+\gamma)^{-2}
\end{aligned}
\end{equation}
with $x(\tau,\theta),y(\tau,\theta)$ as defined in (\ref{scaling1}). If $\tilde \tau_k\to\tau_k$ as $T\to\infty$, then the theorem is proven. This is the case, algebraic computations lead to
\mbox{$\tilde\tau_k=\tau_k+\Or((s_k-\tau_k^2)T^{-1/3})$} as desired.

Consider now the case (b). From (\ref{eqTASEPvsDP}) we have
\begin{equation}\label{eq3.71b}
\Pb\bigg(\bigcap_{k=1}^m\{h_{\theta_k T}(\xi_k\theta_k T)\geq h_{\rm ma}(\xi_k)\theta_k T-2s_k T^{1/2}\}\bigg)=\Pb\bigg(\bigcap_{k=1}^m\{L_2(x_k,y_k)\leq \theta_k T\}\bigg)
\end{equation}
with
\begin{equation}
\begin{aligned}
x_k&=\tfrac12(\xi_k\theta_kT+h_{\rm ma}(\xi_k)\theta_k T-2s_k T^{1/2}),\\
y_k&=\tfrac12(h_{\rm ma}(\xi_k)\theta_k T-\xi_k\theta_k T-2s_k T^{1/2}).
\end{aligned}
\end{equation}
Let us focus on the case $\xi_k\leq 1-(\rho_-+\rho_+)$ (remind $\eta=\rho_-$ and $\pi=\rho_+$); we have
\begin{equation}
\begin{aligned}
x_k&=(1-\rho_-)(\rho_-+\xi_k)\theta_kT+s_k T^{1/2},\\
y_k&=\rho_-(1-\rho_--\xi_k)\theta_kT+s_k T^{1/2}.
\end{aligned}
\end{equation}
Let $\tilde \theta_k$ and $\tilde\gamma_k$ such that $x_k=\tilde\theta_kT$ and $y_k=\tilde\theta_k\tilde\gamma_k^2T$. Then, we get
\begin{equation}
\theta_k T=S(\tilde\gamma_k)\tilde\theta_k T+s_k\frac{T^{1/2}}{\rho_-(1-\rho_-)},
\end{equation}
and
\begin{equation}
\tilde \theta_k=(1-\rho_-)(\rho_-+\xi_k)\theta_k+\Or(T^{-1/2}).
\end{equation}
Then, we can apply directly (\ref{bmeqn}) with $\theta_k$ replaced by $\tilde\theta_k$, $\gamma_k$ by $\tilde\gamma_k$ and $s_k$ replaced by $s_k/(\rho_-(1-\rho_-))$ (similarly for the case $\xi_k>1-(\rho_-+\rho_+)$) to get
\begin{equation}
\begin{aligned}
\lim_{T\to\infty}&\Pb\bigg(\bigcap_{k=1}^m\{L_2(x_k,y_k)\leq \theta_k T\}\bigg) \\
&=\Pb\left(\bigcap_{k=1}^{m_l+m_s}\left\{\mathcal{B}\left(\theta_k\frac{1-2\rho_--\xi_k}{\rho_-(1-\rho_-)}\right)\leq \frac{s_k}{\rho_-(1-\rho_-)}\right\}\right) \\
&\times \Pb\left(\bigcap_{k=m_l+1}^m\left\{\mathcal{B}'\left(\theta_k\frac{\xi_k+2\rho_+-1}{\rho_+(1-\rho_+)}\right)\leq \frac{s_k}{\rho_+(1-\rho_+)}\right\}\right).
\end{aligned}
\end{equation}
Finally, one uses the scaling of Brownian Motion to rewrite it as in (\ref{eq2.12}).
\end{proof}


\end{document}